\documentclass[a4paper,12pt]{article}

\usepackage[utf8]{inputenc}
\usepackage{fontawesome5}
\usepackage{authblk}
\usepackage[english]{babel}
\usepackage{graphicx}
\usepackage{fancyhdr}
\usepackage[svgnames]{xcolor}
\usepackage{listings}
\usepackage{dsfont}
\usepackage{bbm}
\usepackage{blindtext}
\usepackage{titlesec}
\usepackage{algorithm}
\usepackage{algpseudocode}
\usepackage{caption}
\usepackage{amsthm}
\usepackage[round]{natbib}
\usepackage{paralist}
\usepackage{mathtools}
\usepackage{tabularx}

\usepackage{amsmath,amsfonts,amssymb,bm}
\usepackage{booktabs}
\usepackage[hidelinks]{hyperref}

\newtheorem{assumption}{Assumption}[section]

\providecommand{\keywords}[1]
{
  \small	
  \textbf{Keywords:} #1
}

\graphicspath{{Fig/}}
\usepackage[lmargin=2.4cm, rmargin=2.25cm, bmargin=3cm, tmargin=2.5cm]{geometry}
\usepackage{ulem}
\newtheorem{theorem}{Theorem}

\newtheorem{proposition}[theorem]{Proposition}%

\begin{document}


\title{Inference for Diffusion Processes via Controlled Sequential Monte Carlo and Splitting Schemes}
\author[1]{S. Huang}
\author[1]{R. G. Everitt}
\author[1]{M. Tamborrino}
\author[1]{A. M. Johansen}
\affil[1]{\small Department of Statistics, University of Warwick, Coventry,UK}
\maketitle

 \begin{abstract}
We introduce an inferential framework for a wide class of semi-linear stochastic differential equations (SDEs). Recent work has shown that numerical splitting schemes can preserve critical properties of such types of SDEs, 
   give rise to explicit pseudolikelihoods, and hence allow for parameter inference for fully observed processes. Here, under several discrete time observation regimes (particularly, partially and fully observed with and without noise), we represent the implied pseudolikelihood as the normalising constant of a Feynman--Kac flow, allowing its efficient estimation via controlled sequential Monte Carlo and adapt likelihood-based methods to exploit this pseudolikelihood for inference. 
   The strategy developed herein allows us to obtain good inferential results across a range of problems. Using diffusion bridges, we are able to computationally reduce bias coming from time-discretisation without recourse to more complex numerical schemes which typically require considerable application-specific efforts. 
     Simulations illustrate that our method provides an excellent trade-off between computational efficiency and accuracy, under hypoellipticity, for both point and posterior estimation. Application to a neuroscience example shows the good performance of the method in challenging settings.
\end{abstract}

\noindent\keywords{
diffusion bridges, parameter estimation, particle filtering, pseudolikelihood, stochastic differential equation}
\section{Introduction}
Parameter inference for stochastic differential equation (SDE) models is a broad and challenging task which finds application across many disciplines where SDEs play a crucial role, including biology \citep{Ditlevsen2013}, neuroscience \citep{FitzHugh1955}, and finance \citep{Glasserman2003}, among many (see, e.g., \citet[Chapter 7]{kloedenplaten} {and \cite{Iacus}} for diverse examples). 
The fundamental difficulty is that most SDEs do not admit an explicit transition density in a closed form, which complicates likelihood-based inference. A system of equations that mimics, in discrete time, the dynamics of such SDEs, is typically employed. Such approximations are frequently referred to as time-discretisations or numerical schemes, see e.g. \citet[Chapter 9]{kloedenplaten}. With their help, we seek to compute the implied pseudolikelihood as a proxy for the likelihood. Here, further difficulty resides: not all time-discretisations are convergent and able to preserve important model properties, and those that do,  may not give rise to tractable pseudolikelihoods. A systematic approach to the computation/estimation of such likelihoods is therefore needed. Given a mechanism for accessing these pseudolikelihoods, a wide range of standard tools including likelihood maximisation and Markov chain Monte Carlo (MCMC) can then provide inference.

\subsection{Numerical schemes}
The quality of a numerical scheme has a profound impact on inference. We confine our attention to parameter inference for the following class of $d$-dimensional SDEs:
\begin{equation}
\label{classofsde}
    dX_t=F(X_t)dt+\Sigma dW_t=(AX_t+\gamma(X_t))dt+\Sigma dW_t,
\end{equation}
where $F:\mathbb{R}^d \to \mathbb{R}^d$ is the drift coefficient, $A \in \mathbb{R}^{d\times d}$ is a matrix, $\Sigma \in \mathbb{R}^{d\times m}$ is the diffusion coefficient, $\gamma$: $\mathbb{R}^d\to\mathbb{R}^d$ is a function which we assume to be non-linear with polynomial growth, and $W_t$ is an $m$-dimensional Brownian motion. We assume that the SDE starts at $X_0$, whose distribution has bounded moments and a density $\mu$ that is known, independent of the Brownian motion $\{W_t\}_{t\ge 0}$. The simplest and perhaps most intuitive explicit numerical scheme is the Euler--Maruyama (EuM) scheme, which approximates 
{the SDE solution} by fixing the value of the drift and diffusion coefficients over finite time intervals (see, e.g., \citealt[Chapter 10]{kloedenplaten}). However, it was proved in \cite{hutzen} that EuM diverges in finite time for a wide class of SDEs with non-globally Lipschitz coefficients, such as \eqref{classofsde}. Numerical schemes have been developed to address this, many of which do not admit explicit transition densities. In this paper, we explore the use of splitting schemes, a class of explicit schemes recently proposed and analysed  for the class of problems of interest in \cite{buckwar2022splitting}, and find that they enjoy good numerical stability {and property preservation} while having simplicity comparable to that of EuM. {As we will eventually look into hypoelliptic SDEs, i.e. SDEs with a smooth transition density whose diffusion coefficient is degenerate, the numerical preservation of such property (called 1-step hypoellipticity)  is also important, which is a feature enjoyed by the splitting schemes but not EuM \citep{buckwar2022splitting}.} 
\subsection{Observation regimes}
The complexity of inference on SDEs begins to unfold once we start addressing practical concerns. Depending on how the data is collected, a number of observation regimes can arise, each of which requires a different inference strategy.

To introduce the topic precisely, we imagine here a continuous-time diffusion process $\{X_t\}_{t\ge0}$ of dimension $d$, a discrete-time observation process $\{Y_k\}_{k=0}^{M}$ {(with $Y_k:=Y_{t_k}$)} of {dimension} $d_y\le d_x$, some noise terms $\{\epsilon_k\}_{k=0}^M$, taking values in $\mathbb{R}^{d'}$ that are serially independent, an invertible function $\Lambda:\mathbb{R}^{d}\rightarrow\mathbb{R}^{d'}$, and a non-invertible function $\Xi:\mathbb{R}^{d}\rightarrow\mathbb{R}^{d'}$.  The three most widespread regimes are instantaneous observation at discrete time points, either fully or partially:
\begin{compactitem} 
    \item[(i)] full observation without noise, i.e. $Y_k=\Lambda(X_k)$ so one has full information about the process at the observation times ${t_k}$;
    \item[(ii)] partial observation without noise, i.e. $Y_k=\Xi(X_k)$;
    \item[(iii)] full or partial observation with additive observation noise, i.e. $Y_k=\Lambda(X_k)+\epsilon_k$ or $Y_k=\Xi(X_k)+\epsilon_k$.
\end{compactitem}%

In case (iii), where noise term $\epsilon_t$ is present, the structure of the diffusion process and of its observations becomes exactly that of a state space model (SSM). See \cite{inferssm} for a survey of static parameter inference for SSMs. While the methodology presented {here} most straightforwardly applies to SSMs, we also develop methodology applicable to the cases where observations are made without measurement noise, which counter-intuitively,  are more challenging computationally. 

In the noiseless fully-observed case (i), given a splitting scheme of the sort we consider, the pseudolikelihood can be extracted exactly, {e.g. for} $\Lambda$ being the $identity$ map, as  proposed in \cite{pilipovic2022parameter}. We emphasise that {such inference could be further improved, as the} bias arising from time-discretisation can be made arbitrarily small under mild regularity conditions by appealing to a data augmentation method known as diffusion bridge sampling, see \cite{Durham}. Asymptotically, this would lead to \textit{exact} inference {for convergent numerical schemes, as illustrated here in Section~\ref{section5}}. 

Finally, case (ii) is particularly common in applications, in which only some of the coordinates are observed, e.g. for $\Xi$ being a $projection$ map. This is typically the case in, e.g., neuroscience,  when modelling single or population neural activity. We will represent the associated intractable likelihoods as normalising constants that will then be efficiently targeted using sequential Monte Carlo (SMC).
\subsection{Sequential Monte Carlo}
SMC methods were popularised by \cite{Gorden} in the filtering context as the bootstrap particle filter (BPF); a recent overview is provided by \cite{introtosmc}. Serving as an estimation tool for the marginal likelihood of SSMs, SMC will lie at the heart of our proposed methodology. In the partially observed without noise case (ii), 
\cite{jrssb} employed a bootstrap particle filter (and backward smoother) targeting the distribution of the unobserved coordinates, {using} their proposed numerical scheme which has a transition density that lies within the natural exponential family, to perform a Monte Carlo expectation-maximisation type optimisation. We develop a different approach, employing the recently developed iterated/controlled {SMC method} 
for {both} smoothing and likelihood estimation (\citealt{iapf,cSMC}), which we will collectively term controlled sequential Monte Carlo (cSMC), to provide extremely stable estimates of the pseudolikelihoods at small computational cost (noting that \textit{conditional} SMC, often referred to as CSMC, is a different algorithm).
\subsection{Review of alternative approaches}
Inference for SDEs is well-studied, see \cite{SDEreview} for a recent review. We briefly mention here several alternative methods that are not unrelated to our cSMC-based approach, in the sense that they allow asymptotically exact inference. The exact simulation approach developed in \cite{exactsim} outlines how, under very restrictive conditions, using rejection sampling, one may sample exactly from the law of the diffusion process, motivating a host of simulation-based {techniques} with considerable computational costs. 
The manifold MCMC approach developed by \cite{manifoldmcmc} is a Bayesian method that is applicable to diffusion processes across a wide range of settings with the ability to reduce discretisation errors. The method circumvents explicitly extracting the likelihood by augmenting the target space with latent diffusion paths and using a constrained Hamiltonian Monte Carlo method to sample from the joint posterior. The likelihood is thus implicitly handled via the posterior distribution, but not explicitly computed. Pursuing the conditional law of a diffusion given a fixed end point observation gives rise to the class of methods that we will call \textit{diffusion bridges}, see recent advances by \cite{Whitaker2017} and \citet[Chapter 3]{malorythesis} which rely on tailored importance sampling regimes to estimate the improved likelihoods. Further work along this line includes \cite{heng2022simulatingdiffusionbridgesscore} for designing guided time-reversal diffusion bridges,  motivating \cite{boserup2024parameterinferencedifferentiablediffusion} to develop an inference regime again based on importance sampling. The most recent work by \citet[Chapter 2]{iguchithesis} developed an explicit sampling scheme and derived closed-form transition density to compute its approximate likelihood,  motivating a likelihood-based inference approach. This method, while enjoying high accuracy and being robust against hypoellipticity, requires considerable application-specific effort.
\subsection{Overview of the paper}
In Sections~\ref{section2} and \ref{section3}, we provide the necessary background on time-discretisation schemes and cSMC, respectively, the two elements upon which our methodology builds. 
 Section~\ref{section4} {makes three main contributions: it} represents various forms of likelihood as integrals, which we eventually transform into quantities which cSMC can readily estimate; {it} combines the splitting scheme with the derived identities to form a computationally tractable method to estimate the implied pseudolikelihood; and {it} reviews two computational tools suitable for conducting inference using cSMC-approximated pseudolikelihoods. In Section~\ref{section5}, numerical results are presented for a univariate cubic SDE and the {two-dimensional, partially observed hypoelliptic} stochastic FitzHugh--Nagumo (FHN) model. We examine the explosion of the EuM scheme and the bias-reduction effect of our proposed methodology in the former example, and the potential of the method for both frequentist and Bayesian inference in realistic settings in the latter, where both simulated and real data are considered.
\section{Time-discretisation schemes}\label{section2} 
In the scenario we consider, $M+1$ data are collected at fixed and known discrete time points $t_0<t_1<\ldots<t_M=T $. Where there is no ambiguity, e.g. for equidistant time points, we will, here onwards, always use abbreviated subscripts $k$ for $t_k$, e.g. $X_{k}\equiv X_{t_k}$ for discrete data points. Numerical schemes are time-discretisation methods that aim to approximate the true solution of the SDE at such time points (and, hence, the transition density between them). Upon the existence of a unique strong solution of the SDE, we will assume that the process gives rise to a smooth transition density, see e.g. \cite{buckwar2022splitting} and \cite{jrssb} for the discussion of sufficient conditions. The true transition density, i.e. the density of the law of the process at a time $t$ conditional upon its state at an earlier time, say $s<t$, is denoted by $f^{*}_{s,t}(x_t|x_s)$. It typically does not admit an explicit form. Given an approximation ${f}_{s,t}(x_t|x_s)$ which \textit{does} admit an explicit form, it is possible to evaluate the implied pseudolikelihood, using $f_k$ to denote $f_{t_{k-1},t_k}$:
\begin{eqnarray}
\label{pseudolikelihood}
    {f}(x_{{0:M}})&=&\mu(x_0)\prod_{k=1}^{M} {f}_k(x_{{k}}|x_{{k-1}}),
\end{eqnarray}
where $x_{{0:M}}$ denotes the available data $x_{0},\ldots, x_{M}$. 
Section 6 of \cite{pilipovic2022parameter} provides a simulation study in the fully observed case (i) that explores the asymptotic properties of a MLE based on pseudolikelihoods in the form of \eqref{pseudolikelihood} {for the considered class of SDEs \eqref{classofsde} {based on the numerical schemes of \cite{buckwar2022splitting}.}} We aim to improve upon this and develop inference methods for the same class of SDEs {and time-discretisation scheme}, but under much wider observation regimes.

\subsection{Assumptions}
Let $\lVert\cdot\rVert$ and $\langle\cdot,\cdot\rangle$ denote the Euclidean norm and inner product on $\mathbb{R}^d$, respectively, and let $\mathcal{C}^2(\mathbb{R}^d)$ denote the class of twice-continuously-differentiable functions from $\mathbb{R}^d$ to $\mathbb{R}$. For the stated class of SDE \eqref{classofsde}, we impose the following:
\begin{assumption}
\label{assumption1}
There exist constants $c_1$, $c_2>0$ and $\chi\ge1$, such that\textcolor{brown}{,} for all $x,y\in\mathbb{R}^d$\textcolor{brown}{,} the non-linear function $\gamma\in\mathcal{C}^2(\mathbb{R}^d)$ satisfies:
\begin{eqnarray*}
   &(a)& \langle \gamma(x) - \gamma(y), x-y \rangle\le c_1\lVert x-y\rVert^2,\\
   &(b)& \lVert \gamma(x)-\gamma(y)\rVert^2\le c_2(1+\lVert x\rVert^{2\chi-2}+\lVert y\rVert^{2\chi-2})\lVert x-y\rVert^2.
\end{eqnarray*}
\end{assumption}
This is an instance of a locally Lipschitz drift, with constraints to control its growth so that the existence of a unique strong solution is guaranteed and the moments are finite, see discussions in, e.g., \cite{strongsolution}. This is a weaker assumption on the drift than the commonly-imposed global Lipschitz condition \citep[Part 2]{kloedenplaten} often used to control its growth.  
\subsection{Euler--Maruyama}
For ease of representation, we consider equidistant discrete times with step size $\Delta$, i.e. $t_k = k \Delta$.  The EuM 
approximation of \eqref{classofsde} is defined as:
\begin{equation}\label{EuM}
X^{\textrm{EuM}}_{k}=X^{\textrm{EuM}}_{{k-1}}+\Delta F(X^{\textrm{EuM}}_{{k-1}})+\xi^{\textrm{EuM}}_{\Delta,k},\qquad\xi^{\textrm{EuM}}_{\Delta,k}\overset{\text{iid}}{\sim}N(0,\Delta\Sigma\Sigma^T).
\end{equation}In the case of SDEs with \textit{globally} Lipschitz drift and diffusion coefficients, which is a stronger condition than Assumption~\ref{assumption1}, EuM  
is a well-known instance of a strong Taylor approximation with conditionally Gaussian transition densities. The strong convergence  in the $L_2$ sense will effectively prevent divergence of the simulated path $X^{\textrm{EuM}}_k$ at time $t_k$ from the true solution $X_k$, so that we have for all $0\le k\le M$, $\mathbb{E}[||X^{\textrm{EuM}}_k-X_k||^2]^{\frac{1}{2}}\to0$ as $ \Delta\to0$. For non-globally Lipschitz SDEs of the form \eqref{classofsde}, Theorem 2.1 in \cite{hutzen} {gives} a sufficient condition under which the EuM approximation can be proven to be $L_p$ divergent in finite time, in one-dimension. This will be an event happening with increasing probability as the diffusion parameter increases. When this happens, simulation-based methods for inference will be prohibited to various extents depending on the specific approach.

For our stated class of diffusion processes, we require \textit{strongly} convergent {numerical} schemes admitting an explicit transition density which is needed to implement cSMC. To this end, in this work, we advocate the use of splitting schemes.
\subsection{Splitting schemes}
Splitting schemes decompose the initial SDE, here \eqref{classofsde}, into solvable sub-equations, and then compose the resulting solutions iteratively over (short) finite time intervals. Following \cite{buckwar2022splitting}, we split \eqref{classofsde} into two sub-equations:

\begin{tabularx}{\textwidth}{XX}
\begin{equation}
dX^{[1]}_t=AX^{[1]}_tdt+\Sigma dW_t \label{subeq1}
\end{equation}%
&
\begin{equation}
dX^{[2]}_t=\gamma(X^{[2]}_t)dt.\label{subeq2}
\end{equation}%
\end{tabularx}%
\newline
\noindent Equation~\eqref{subeq1} is ``linear in the narrow sense'',  with explicit solution:
\begin{equation}\label{subeq1sol}
X^{[1]}_{k}=e^{A\Delta}X^{[1]}_{{k-1}}+\xi_{\Delta,k}, \qquad
\xi_{\Delta,k}\overset{\text{iid}}{\sim} N\Big(0,\underset{=:C(\Delta)}{\underbrace{\int_0^{\Delta} e^{A(\Delta-s)}\Sigma\Sigma^\intercal(e^{A(\Delta-s)})^\intercal ds}}\Big),
\end{equation}
{where $e^{A\Delta}$ and $C(\Delta)$ denote the matrix exponential and the covariance matrix, respectively. }
Equation (\ref{subeq2}), involving $X_{k}^{[2]}$, is a vector-valued ordinary differential equation (ODE) { which admits a global solution (i.e., it does not explode in finite time) owing to Assumption \eqref{assumption1}, see \cite{HumphriesStuart2002}.} We denote by $\Gamma_{\Delta}:\mathbb{R}^d\to\mathbb{R}^d$ the exact solution of the ODE  evaluated at $t_k$ starting from $t_{k-1}$, i.e. $X^{[2]}_{k}=\Gamma_\Delta(X^{[2]}_{{k-1}})$. If the solution is not available, numerical methods for locally Lipschitz ODEs (e.g., \citealt{Haireretal2000} and \citealt{HumphriesStuart2002}) could be used to approximate it.

\noindent{\bf Lie--Trotter splitting} 
For our semi-linear class of SDE, the {considered} Lie--Trotter scheme composes the solutions $X^{[1]}_{k}$ and $X^{[2]}_{k}$ as (see \cite{buckwar2022splitting}):
\begin{equation}
\label{LT}
    X_{k}^{\textrm{LT}}=
e^{A\Delta}\Gamma_\Delta(X^{\textrm{LT}}_{{k-1}})+\xi_{\Delta,k}.
\end{equation}
We observe crucially that the resulting process is a Markov process, and that the transition density is Gaussian. The resulting numerical scheme, under stated assumptions, is mean-square convergent of order 1, see, e.g., Theorem 2 in \cite{buckwar2022splitting}. This is a property that the EuM demonstrably does not typically enjoy for the considered {locally Lipschitz drift coefficients} satisfying Assumption~\ref{assumption1}, see \cite{hutzen} in the context of forward sampling, and Section~\ref{cubicsde_cs} in {this paper in} the context of diffusion bridge sampling. 

\noindent{\bf Strang splitting} The Strang splitting method employs the same decomposition of the SDE, but differs from  Lie--Trotter in the composition of the solutions to the component differential equations. Here, we consider (see \citealt{buckwar2022splitting}):
\begin{equation}
\label{S}
    X_{k}^{\textrm{S}}= \Gamma_\frac{\Delta}{2}\left(e^{A\Delta}\Gamma_\frac{\Delta}{2}(X^{\textrm{S}}_{{k-1}})+\xi_{\Delta,k}\right),
\end{equation}
which sandwiches solution $X^{[1]}_{k}$ in the middle of two half-steps of solution $X^{[2]}_{k}$. Again, this is a Markov process, and has the same strong order of $L_2$ approximation error as the Lie--Trotter scheme, see Theorem 3.7 \cite{pilipovic2022parameter}. 

It is particularly attractive for inference that the Strang scheme produces a more accurate one-step approximation than Lie--Trotter \citep[Propositions 3.4 and 3.6]{pilipovic2022parameter}.  Recalling that $\Delta=t_k-t_{k-1}$, and defining $X_k^{\textrm{S}}=\Phi^{\textrm{S}}_{\Delta}(X_{{k-1}})$ and $X_k^{\textrm{LT}}=\Phi^{\textrm{LT}}_{\Delta}(X_{{k-1}})$, for all $x\in \mathbb{R}^d$, there it was established that
$\lVert \mathbb{E}[X_{k}-\Phi^{\textrm{S}}_{\Delta}(X_{{k-1}})|X_{{k-1}}=x]\rVert=\mathcal{O}(\Delta^3), \lVert \mathbb{E}[X_{k}-\Phi^{\textrm{LT}}_{\Delta}(X_{{k-1}})|X_{{k-1}}=x]\rVert=\mathcal{O}(\Delta^2)$. These exponents characterise the \textit{one-step weak order} of a numerical scheme. For Lie--Trotter and Strang, this order is $2$ and $3$, respectively. This means that, for a small enough  time step $\Delta$, the average one-step discrepancy between the true process $X_{k}$ and the Strang-approximated process $\Phi^{\textrm{S}}_{\Delta}(X_{{k-1}})$ {is smaller} than that arising from the Lie--Trotter approximation $\Phi^{\textrm{LT}}_{\Delta}(X_{{k-1}})$, for any given previous state $X_{{k-1}}=x$. The pseudolikelihood being a product of transition densities due to the Markov property, will hence on average have a smaller approximation error. 

A complication with using Strang for inference, is that its transition density is not Gaussian. However, it can be expressed as a transformation of a Gaussian random variable and the previous state. In particular, {if the solution of the ODE subequation \eqref{subeq2} is invertible,} the inverse under $\Gamma_\frac{\Delta}{2}$ of $X_{k}$ is given by:
\begin{equation}
\label{eqcond}
    \Gamma_\frac{\Delta}{2}^{-1}(X_{k}^{\textrm{S}})=e^{A\Delta}\Gamma_\Delta \left( \Gamma_\frac{\Delta}{2}^{-1}(X^{\textrm{S}}_{{k-1}})\right)+\xi_{\Delta,k}.
\end{equation}
With this change of variable, \cite{pilipovic2022parameter} obtained the pseudolikelihood and developed MLE inference for fully-observed processes (i). {As the solution $\Gamma_\Delta$ may not be invertible, the Strang scheme is less widely applicable to inference than Lie--Trotter.
  We address this at the end of Section~\ref{non-invertible}.   
 \section{Controlled sequential Monte Carlo}\label{section3}
Viewing the intractable likelihood as the normalising constant of a particular class of models allows SMC methods to provide unbiased estimates of the pseudolikelihood. We will preview the case of observation regime (ii) in this section, before introducing diffusion bridges and revisiting the other regimes in Section~\ref{section4}.

Consider the full process $x_k=(v_k,u_k)$ taking values on $\mathcal{X}=\mathcal{V}\times\mathcal{U}$, where only $\{v_k\}_{k=0}^{M}$ is observed. We want to extract the marginal density of  $v_{0:M}$:
\begin{eqnarray}\label{partial_og}
    f(v_{0:M})&=&\int_{\mathcal{U}^{M+1}}\mu(v_0,u_0)\prod_{k=1}^{M}f_k(v_k,u_k\mid v_{k-1},u_{k-1})du_{0:M},
\end{eqnarray}
where $f$ denotes the partial pseudolikelihood arising from a numerical scheme. We use $f_k$ to denote its transition density at time $t_k$ and $\mu$ some initial distribution. 

Now, let the implied transition density $f_k$  be a Gaussian distribution over $(v_k,u_k)$ conditionally on $(v_{k-1},u_{k-1})$ which can be decomposed, conditional upon $(v_{k-1},u_{k-1})$, as the marginal density of $v_k$ and the conditional of $u_k$ given $v_k$: 
\begin{equation}
\label{decomp}
    f_k(v_k,u_k|v_{k-1},u_{k-1})=f_k^1(v_k|v_{k-1},u_{k-1}) f_k^2(u_k|v_k,u_{k-1},v_{k-1}), 
\end{equation}
where both $f_k^1$ and $f_k^2$ are analytically tractable.

Lie--Trotter splitting fits into this class. Under assumptions detailed in Section~\ref{non-invertible}, 
also the Strang-implied transition density will fit into this class. Then, we write the marginal pseudolikelihood of the first coordinate $v_{0:M}$ as an integral:
\begin{equation}\label{partial_psedo_unbridged}
f(v_{0:M})=\int_{\mathcal{U}^{M+1}}\mu^1(v_0) \mu^2(u_0|v_0) \prod_{k=1}^{M} f_k^1(v_{k}|v_{k-1},u_{k-1})f_k^2(u_k|v_k,u_{k-1},v_{k-1})du_{0:{M}},
\end{equation}
where $\mu^1$ and $\mu^2$ are the marginal and conditional distribution of the initial distribution $\mu$, respectively. 
For a general system satisfying the factorisation \eqref{decomp}, there is no analytic solution for the $M+1$-fold integration \eqref{partial_og}, i.e. we have an intractable likelihood problem. This section will introduce a class of problems for which SMC is particularly effective, show that \eqref{partial_og} {(and other obtained pseudolikelihoods for all considered observation regimes)} can be transformed into such a problem, recall the basic BPF as a tool for estimating $f(v_{0:M})$ and finally present the cSMC approach to more efficiently estimate this crucial quantity. Once this is done, the obtained pseudolikelihood can then be numerically maximised when performing MLE, or used to perform Bayesian inference, as described in Section~\ref{sec:bayesianinference}.
\subsection{Background}\label{sec:background}
Consider the path of a discrete time Markov process taking values in $\mathcal{X}$ over times $t_0,\ldots,t_M$ operating on $\mathcal{X}^{M+1}$ with a density defined as:
\begin{equation*}
    \mathbb{M}(x_{0:M})=M_0(x_0)\prod_{k=1}^M M_k(x_{k-1},x_k),
\end{equation*}
where each $M_k$ is a smooth transition density {from $x_{k-1}$ to $x_k$} at time $t_k$. Now, with a sequence of non-negative functions $\{G_k:\mathcal{X}\to[0,\infty)\}_{k=0}^M$, whose product is integrable with respect to $\mathbb{M}(x_{0:M})$, we define: 
\begin{equation}
\label{generaltargetmeasure}
Z=\int_{\mathcal{X}^{M+1}} G_0(x_0)M_0({x_0})\prod_{k=1}^M G_k(x_k)M_k(x_{k-1},x_k)dx_{0:M}. 
\end{equation}
The sequence of functions $\{G_k\}_{k=0}^M$ are typically called \textit{potential} functions. The pair $\{G_k, M_k\}_{k=0}^{M}$ will be termed a \textit{Feynman--Kac formulation} below. Through \eqref{generaltargetmeasure}, for each such pair a new process is defined on $\mathcal{X}$, with a normalised density $\mathbb{Q}(x_{0:M})=1/Z\prod_{k=0}^{M}G_k(x_k)\mathbb{M}(x_{0:M})$, which we refer to as the \textit{target} process with a normalising constant $Z$. We will ultimately represent the solution to problems like \eqref{partial_psedo_unbridged} in the form (13), so that $Z$, will be the likelihood of interest.

\subsection{Feynman--Kac formulation of marginal pseudolikelihood}
Let us go back to our initial problem of {computing the marginal} pseudolikelihood as formulated by \eqref{partial_psedo_unbridged}. Seeing $\{v_k\}_{k=0}^{M}$ as constants (hence omitting them from notation) and $\{u_k\}_{k=0}^{M}$ as the underlying Markov path process, we are interested in re-writing \eqref{partial_psedo_unbridged} as a target density in the form of \eqref{generaltargetmeasure}, whose normalising constant, $Z$, is the marginal pseudolikelihood. To this end, we have the following proposition, with proof reported in Supplementary Material \ref{Proof}.
\begin{proposition}\label{prop1}
The Feynman--Kac formulation, for $k=1, \ldots, M-1$:
\begin{align}  \label{prop1eq1}G_0(u_0)=&\mu^1(v_0)f_1^1(v_1|v_{0},u_{0}) & M_0(u_0)=&\mu^2(u_0|v_0)\\
    \label{prop1eq2}G_k(u_k)=&f_{k+1}^1(v_{k+1}|v_{k},u_{k}) & M_k(u_{k-1},u_k)=&f_k^2(u_k|v_k,u_{k-1},v_{k-1}),
\end{align}
has normalising constant (see~\eqref{generaltargetmeasure}) $Z=f(v_{0:M})${, the marginal pseudolikelihood}.
\end{proposition}
The {cSMC} method while also characterising the conditional distribution of the latent process, can be viewed as a mechanism for efficiently estimating $Z$, the likelihood. To make the algorithm digestible, we first introduce a class of equivalent models, in the sense that they have the same normalising constant and target density.
\subsection{Twisted models}
Consider a sequence {of strictly positive and bounded functions $\{\psi_k: \mathcal{X}\to (0,\infty)\}_{k=0}^M$.}
These functions are referred to as {\textit{policies}} in \cite{cSMC}, as they can be viewed as a control input which will be used to perturb the dynamics of an underlying Markov process with the particular goal of providing low variance estimates of a particular integral of the type \eqref{generaltargetmeasure}.  
We will adopt the following short hand: $
        M_k(\psi_k)(x_{k-1}):=\int_{\mathcal{X}}\psi_k(x_k)M_k(x_{k-1},x_k)dx_k$, for any $\psi_k : \mathcal{X} \to {(0,\infty)}$.
        
We will frequently use $[a:b] := [a,b] \cap \mathbb{N}_0$ to denote the discrete interval \textit{from} $a$ \textit{to} $b$. For each integral specified by the Feynman--Kac formulation $\{G_k,M_k\}_{k=0}^{M}$ , we define the following \textit{twisted} Feynman--Kac formulation which leaves the normalising constant $Z$ unchanged:
\begin{align}
    M_0^{\psi}(x_0)=&\frac{\psi_0(x_0)M_0(x_0)}{M_0(\psi_0)}&G^{\psi}_0(x_0)=&\frac{M_0(\psi_0)G_0(x_0)M_1(\psi_1)(x_0)}{\psi_0(x_0)}\\
    M_k^{\psi}(x_{k-1},x_k)=&\frac{\psi_k(x_k)M_k(x_{k-1},x_k)}{M_k(\psi_k)(x_{k-1})}&G_k^{\psi}(x_k)=&\frac{G_k(x_k)M_{k+1}(\psi_{k+1})(x_{k})}{\psi_k(x_k)}
\end{align}
where $k\in [1:M]$, $\psi_{M+1}\equiv 1$, and the twisted process path has density:
\begin{equation}
    \mathbb{M}^{\psi}(x_{0:M})=M_0^{\psi}(x_0)\prod_{k=1}^M M_k^{\psi}(x_{k-1},x_k).
\end{equation} That is, for the \textit{twisted} target process whose density $\mathbb{Q}^\psi$ {is} defined through 
$\frac{1}{Z^\psi}\prod_{k=1}^M G^{\psi}_k(x_k)\mathbb{M}^{\psi}(x_{0:M})$, we have $\mathbb{Q}^{\psi}(x_{0:M})=\mathbb{Q}(x_{0:M})=\frac{1}{Z}\prod_{k=1}^M G(x_k)\mathbb{M}(x_{0:M})$ and both densities have the same normalising constant, i.e. $Z^\psi=Z$ (see, e.g., \citealt[Proposition 1]{iapf}). The intuition motivating cSMC is that by pulling the density of the twisted path process $\mathbb{M}^{\psi}(x_{0:M})$ closer to the density of the target process $\mathbb{Q}(x_{0:M})$ via choosing an appropriate sequence of policies $\psi_k$, and adapting the \textit{twisted} potentials $\{G^{\psi}_k\}_{k=0}^M$ accordingly, the density of the \textit{twisted} target process $\mathbb{Q}^\psi$ is left unchanged, but the variance of the estimate of $Z$ will be reduced. The estimate is unbiased \cite[Proposition 7.4.1]{delmoral}, so smaller variance implies smaller mean-squared error.
\subsection{A simple particle filter for twisted models}
Particle filters are conventionally viewed as the mean-field interacting particle approximation of the target process whose density is $\mathbb{Q}(x_{0:M})$; we, however, mainly view particle filters as an apparatus to provide unbiased estimates of its normalising constant $Z$. Resampling is a key step in SMC algorithms in which one draws from a weighted sample to obtain an unweighted sample suitable for approximating the same distribution. The particle filter in Algorithm~\ref{alg:pf_simple} details the simple multinomial scheme which is used in our empirical study below, but any unbiased scheme (in the sense of \citealt{Pmcmc}) would work, and lower variance schemes would be expected to lead to better empirical performance. As resampling is relatively rare in cSMC, we do not dwell on this point here.

Applying Algorithm~\ref{alg:pf_simple} to the Feynman--Kac Formulation derived in Proposition~\ref{prop1} with constant policies $\psi_{0:M}\equiv 1$ yields the BPF estimate of the partial pseudolikelihood. Such estimates will have a variance that can be improved by more careful selection of the policies{, as discussed in  Section \ref{policyrefinement}. Before doing that though, in the next Section, we recall the possibility of an optimal policy.} 

\begin{algorithm}[t]
\caption{Simple particle filter for the twisted model}\label{alg:pf_simple}
\textbf{Input}: \# particles N, formulation $\{G_k,M_k\}_{k=0}^M$, policy $\psi_{0:M}$.
\begin{algorithmic}

\State \textbf{for} {$n\in[1:N]$}: sample $X_0^n\sim M_0^\psi(x_0)$ and compute  weights $w_0^n=G^\psi_0(X_0^n)$

\State compute normalised weights $W_0^n=\frac{w_0^n}{\sum_{n'=1}^N w_0^{n'}}, n \in [1:N]$ and $l_0=\frac{1}{N}\sum_{n=1}^N w_0^n$
\State \textbf{for} {$k\in[1:M]$} \textbf{do}:
\State $\qquad$ \textbf{if} $\frac{1}{\sum_{n=1}^{N}(W^n)^2} \le \frac{N}{2}$ resample: $A_k \sim \mathsf{Cat}(W_k^{n-1})^{\otimes N}$; set $\widehat{W}_{k-1}^n = (\frac{1}{N},\ldots,\frac{1}{N})$
\State $\qquad$ \textbf{else} $A_{k} = [1:N]$ and $\widehat{W}_{k-1}^n = W_{k-1}^n$
\State $\qquad$ \textbf{for} {$n\in[1:N]$}: sample  $X_k^n\sim M_k^\psi(X_{k-1}^{A_k^n},x_k)$ and set $w_k^n= G_k^\psi(X_{k-1}^{A_k^n},X_k^n)$
\State $\qquad$ compute $W_k^n=\frac{\widehat{W}_{k-1}^n w_k^n}{\sum_{n'=1}^N \widehat{W}_k^{n'-1} w_k^{n'}}, n \in [1:N]$ and $l_k^N= \sum_{n=1}^N \widehat{W}_{k-1}^n w_k^n$
\State \textbf{end for}
\end{algorithmic}
\textbf{Output}: Likelihood estimate $L_N=\prod_{k=0}^M l_k^N$ , particles {$X_{0:M}^{1:N}$}
\end{algorithm}
\subsection{Optimal policy and policy refinement}\label{policyrefinement}
The existence of an optimal policy, under which the particle filter could estimate the normalising constant $Z$ exactly, is well known. 
In particular, given a pair $\{G_k,M_k\}_{k=0}^M$ specifying the target process with density $\mathbb{Q}$ and a normalising constant $Z$, the sequence of optimal policies $\{\psi_k^*\}_{k=0}^M$ satisfies the backward recursion $\psi^*_M(x_M) = G_M(x_M),  \psi^*_k(x_k)=G_{k}(x_k)M_{k+1}(\psi^*_{k+1})(x_k)$ for $k\in[0:M-1]$. The resulting twisted Feynman--Kac formulation $\{G^{\psi^*}_k,M^{\psi^*}_k \}$, when processed with Algorithm~\ref{alg:pf_simple}, gives the estimator $L_N = Z$ almost surely for all $N\in\mathbb{N}$, see e.g. \citet[Proposition 2]{iapf}.

The optimal policy is intractable in general{, but its recursive equation could still be used to refine the policies $\psi_k$.} 
For easy tractability, following \cite{iapf} and \cite{cSMC}, as we are interested in conditionally Gaussian transition kernels $M_k$, we constrain ourselves to
a family of log-policies, that we denote as $\phi_k=\log(\psi_k)$, which when used to twist a Gaussian transition kernel induce a new Gaussian transition kernel:
\begin{equation}\label{f^eq}
    \mathcal{F}_k^{eq}=\{\phi_k(x_k)=x_k^{T}A_kx_k+x_k^Tb_k+c_k:(A_k,b_k,c_k)\in\mathbb{S}_d\times\mathbb{R}^d\times\mathbb{R}\},
\end{equation}
where $d$ is the dimension of the particles and $\mathbb{S}_d$ is the cone of symmetric positive semi-definite matrices of dimension $d\times d$.

As a conjugacy property that is frequently seen in Bayesian computation, the above constraints on $\{\psi_k\}_{k=0}^{M}$ and $\{M_k\}_{k=0}^{M}$ 
conveniently leads to an explicit expression for computing $M_k(\psi_k)(x_{k-1})$-like integrals; sampling $x_k$ at time $t_k$ from $M_k^{\psi}(x_{k-1},x_k)\propto{\psi_k(x_k)M_k(x_{k-1},x_k)}$ involves only sampling from another Gaussian distribution whose mean and variance are easily computed.

To learn the optimal policies in an iterative way, \cite{cSMC} appealed to an iterative policy refinement process known as approximate dynamic programming (ADP), and learnt the optimal policy by iteratively learning incremental refinements of the current policy. The iterated auxiliary particle filter (iAPF) by \cite{iapf}, employed for its numerical examples a more na{\"\i}ve direct numerical optimisation of the policy. These two fitting methods constrained to the family of quadratic log-policies produces very similar results, but are different schematically. We adopt a hybrid form of the two, in the sense that we use ADP to learn the log-policy, but do it by directly learning a policy at each iteration rather than incrementally refining it, as summarised in Algorithm~\ref{alg::adp}.

\begin{algorithm}[t]
\caption{Linear Regression Approximation of the Optimal Policy}\label{alg::adp}
\textbf{Input}: Particles from Algorithm~\ref{alg:pf_simple} $X_{0:M}^{1:N}$, Feynman--Kac formulation $\{G_k,M_k\}$
\begin{algorithmic}
  \For{$k \in [M:0]$}
     \State \textbf{for} $n \in [1:N]$: compute $\xi_k^n=\log(G_k(X^n_k)) + \mathbb{I}_{k < M} \log( M_{k+1}(\hat{\psi}_{k+1})(X_k^{n}))$.
     \State \begin{equation}\label{adp_fitting}
       \textrm{compute }
\hat{\phi}_k=\underset{\phi_k\in\mathcal{F}^{eq}_k}{\text{argmin }} 
 \sum_{n=1}^N (\phi(X_k^n)-\xi_k^n)^2\quad, \textrm{and }\hat{\psi}_k=\exp(\hat{\phi}_k)
     \end{equation}
     \EndFor
\end{algorithmic}
\textbf{Output}: the fitted policy $\hat{\psi}_{0:M}$ as an approximation for the optimal policy 
\end{algorithm}
Under the prespecified constraints, equation \eqref{adp_fitting} is{ a} least squares problem with a total of $(1+d)(d/2+1)$ coefficients. Running Algorithm~\ref{alg:pf_simple} on the resulting {fitted} one-step twisted Feynman--Kac formulation $\{G^{\hat{\psi}}_k,M^{\hat{\psi}}_k\}_{k=0}^{M}$ delivers an unbiased estimate of the normalising constant as the twisted model \citep{iapf}. If the optimal policy has been approximated well, these estimates will have much smaller variance than those provided by a BPF. It is reasonable to suppose that this will be the case in problems of interest, noting that, under regularity conditions, the fitted policies from Algorithm~\ref{alg::adp}, iteratively, converge to a \textit{good} policy, whose characterisation can be found in \citet[Section 5.2, Theorem 1]{cSMC}. 

 In our implementation, we will then simply supply the resulting particles from Algorithm~\ref{alg:pf_simple}, say $X_{0:M}^{1:N}$, again into the recursive function fitting procedure {from} Algorithm~\ref{alg::adp}, {with the same Feynman--Kac formulation as initial input.} Figure~\ref{fig:csmcflow} summarises our implementation of the cSMC algorithm. Supplying the Feynman--Kac formulation derived in Proposition~\ref{prop1} as the functional input into the cSMC, will allow us to solve the intractable likelihood problem (12), estimating the desired marginal pseudolikelihood. The proposed cSMC approach will be applied to all considered observation regimes, with corresponding Feynman--Kac formulations derived in the next Section. Note the final run of Algorithm~\ref{alg:pf_simple} upon termination, to retain unbiasedness for the estimate. 
 \begin{figure}
     \centering     \includegraphics[width=.72\textwidth]{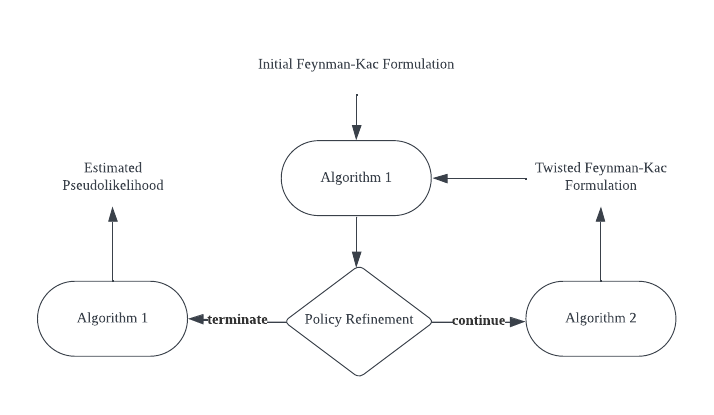}
\caption{(Simplified) Flowchart Diagram of cSMC.}
     \label{fig:csmcflow}
 \end{figure}
\section{Methodology}\label{section4}
Before providing our proposed methodology, we briefly summarise relevant earlier literature: \cite{pedersen1995} gave theoretical guarantees for the use of pseudolikelihoods in MLE-type inferences for a very general class of SDEs, if one is able to device a method (in his case, based on bridges, like \eqref{bridge_kran}) to asymptotically correct the bias arising from the time-discretisation; \cite{Durham} later on improved upon the Pedersen's method \eqref{bridge_kran} by designing a manually twisted diffusion bridge called the modified Brownian bridge, and used importance sampling to prove in concept that it is indeed possible to estimate and maximise the ``corrected'' likelihood, calling for advanced computational approaches and more intricately designed bridges, to reduce the noise coming from such estimation and hence ease the maximisation process. The cSMC approach gives an automatic way to guide the forward proposals within diffusion bridges by starting from the ``blind'' forward proposals of Pedersen's method \eqref{bridge_kran}, sparing the manual efforts to design sufficiently guided forward proposals. For further work, see \cite{Whitaker2017,malorythesis}, which facilitates variance-reduction within importance samplers. \citet[Part IV]{pierthesis} followed this path and appealed to the then novel iAPF {(\citealt{iapf})} to estimate the EuM implied pseudolikelihood for several simple examples. This showed great promise, as the iAPF was indeed able to infer  the pseudolikelihoods in the fully observed with no additive noise case,  with clear advantages in terms of the variance of the likelihood estimates compared with the more established approach of \cite{Durham}. Here, by combining these ideas with more advanced numerical and computational techniques, to improve its efficiency in a non-trivial way, we show that this approach extends much further, across different observation regimes for a wide class of SDEs. 

In all considered settings, we will employ the following time scale and indexing to keep notation simple (although the approach can be applied directly to any time scale with trivial modifications). We introduce equidistant time points $\{t_k = k \Delta\}_{k=0}^M$, with $t_M= M \Delta = T$. Where augmentation is involved, we divide each interval $[t_{k-1},t_k]$, into subintervals. To this end, we introduce for each $k$, the equidistant time points $\{t_{k;n}\}_{n=0}^K$ such that $t_{k-1}=t_{k;0}<t_{k;1}<\ldots<t_{k;K}=t_k$, each being $\delta=\Delta/K$ apart. We will use $k;0:K$, to index the collection of ordered time points, e.g. $t_{k;0:K}$. 
\subsection{Full observation}\label{section:Full_Observation}
\noindent{\bf Baseline result: Unbridged pseudolikelihood} 
In the fully observed case, for a numerical scheme with explicit transition density, the pseudolikelihood can be computed explicitly:
\begin{equation}\label{full}
    f(x_{0:M})=\mu (x_0)
\prod_{k=1}^M f_k\left(
    x_k
\mid
    x_{k-1}
\right),
\end{equation}
where $f_k$ is the implied pseudo-transition density at time $t_k$, $\mu$ some initial distribution, and $f$ the \textit{pseudolikelihood}. Parameter dependence is suppressed in the notation for the likelihood for parsimony. In a high-frequency regime with a convergent numerical scheme, this likelihood will be able to drive ``almost exact'' inference, see \citet[Section 6.5]{pilipovic2022parameter} for a simulation study on the effect of time scale on the quality of likelihood-driven inference. Even in low-frequency regimes, we can obtain results with good theoretical properties asymptotically by computationally reducing the bias via diffusion bridging{, as discussed below}. 

\noindent{\bf Exploratory steps: Bridged full pseudolikelihood}
The logic behind diffusion bridges is straightforward: For each of the proxy transition densities $f_k$ that has a time step, say $\Delta$,  we replace it with a version that is obtained from augmenting with $K$ unobserved intermediate time points which are then integrated out: 
\begin{equation}\label{bridge_kran}
    {f}_{k}^{[K]}(x_k|x_{k-1})=\int_{\mathcal{X}^{K-1}}\prod_{n=1}^{K} f_{k;{n}}(x_{k;{n}}\mid x_{k;{n-1}}) dx_{k;{1:K-1}},\qquad k\in[1:M].
\end{equation} 
The bridged pseudolikelihood is then:
\begin{equation}
\label{fullbridge}
    {f}^{[K]}(x_{0:M})=\mu(x_0)  \prod_{k=1}^M {f}_{k}^{[K]}(x_k|x_{k-1}).
\end{equation}
If the integrals in \eqref{bridge_kran} can be evaluated exactly, which is rarely the case, the more accurate bridged pseudolikelihood \eqref{fullbridge} will enjoy a discretisation error on the size of $\delta$ instead of $\Delta$. Henceforth, we will assume that the following strong convergence of the bridged transition density is met:
\begin{assumption}\label{assumption_conv_likelihood}
Let $\tilde{X}_k$ and $X_k$ denote the discretised and continuous model at time $t_k$, respectively, and $f^*(x_k|x_{k-1})$ denote the true transition density of the continuous model. We assume the following $L_1$ convergence to the true transition density holds for the bridged {transition density $ f_k^{[K]}$ in \eqref{bridge_kran}} for all $k\ge0$ and $x_{k-1}\in\mathcal{X}$:  
$
        \mathbb{E}[|{f}_{k}^{[K]}(\tilde{X}_k|x_{k-1})- f^*(X_k|x_{k-1})|]\to0$ as $K\to\infty.$
\end{assumption}
In particular, the pointwise convergence in probability of the (log-)likelihood can be guaranteed:  $\log(  {f}^{[K]}(\tilde{X}_{0:M}))\to \log(f^{*}(X_{0:M}))$ as $K\to \infty$, with $f^*(X_{0:M})$ being the likelihood of the continuous model. \citet[Theorem 2]{pedersen1995} established consistency and asymptotic normality of the MLE 
based on this bridged likelihood \eqref{fullbridge} under regularity conditions that covers our considered class of SDEs \eqref{classofsde}, and Assumption~\ref{assumption_conv_likelihood}. In general, \eqref{bridge_kran} is analytically intractable, but the following result, whose proof follows by factorising $Z=f^{[K]}(x_k|x_{k-1})$ into the integral form \eqref{generaltargetmeasure} specified by pairs of $\{G_{k;n},M_{k;n}\}_{n=1}^{K-1}$, allows for unbiased cSMC approximations.
\begin{proposition}\label{propbridge_1}
The Feynman--Kac segment, defined for $k\in [1:M]$:
\begin{eqnarray*}
        &&G_{k;n}(x_{k;n})\equiv1, \quad n\in [1:K-2], \qquad\qquad
        G_{k;K-1}(x_{k;{K-1}})= f_{k;{K}}(x_{k}\mid x_{k;{K-1}})\\
        &&M_k(x_{k;{n-1}},x_{k;n})=f_{k;{n}}(x_{k;{n}}\mid x_{k;{n-1}}), \quad n\in [1:K-1]
    \end{eqnarray*}
has normalising constant $f_k^{[K]}(x_k|x_{k-1})$, the bridged full transition density, with given $x_{k;0}=x_{k-1}$ and $x_{k;K}=x_k$. Upon concatenating such M sugments with known initialisation $\mu_0$, the normalising constant of the resulting Feynman--Kac formulation is the bridged full pseudolikelihood $f^{[K]}(x_{0:M})$.\end{proposition}
This construct was referred to as the Pedersen's method by \cite{Durham}. The intermediate sample paths are weighted by the conditional distribution of the $next$ observation given the $last$ in-between simulations, i.e. $f_{k;{K}}(x_{k}\mid x_{k;{K-1}})$. As demonstrated by \citet[Part IV]{pierthesis}, this formulation, when employed within the iAPF with the EuM scheme, 
{yields} a clear variance reduction effect {for estimating the bridged pseudolikelihood} when compared with other methods.   
\subsection{Bridged partial pseudolikelihood}
The partial observation case is more complicated than the fully observed case. The introduction of bridges adds one more layer of complexity. We inherit the assumptions and notation from Proposition~\ref{prop1}, but switch to the finer time scale $\{t_{k;n}\}_{n=0}^K$, with {time step} $\delta$. Consider the following integral:
\begin{align*}
{f}^{[K]}(v_{0:M})=&\int_{\mathcal{U}^{M+1}}\int_{\mathcal{X}^{{M}\times{K}}}\mu^1(v_{0}) \mu^2(u_{0}|v_{0})\prod_{k=1}^{M} f_{k;K}^1(v_{k;K}|x_{{k};{K-1}})f_{k;K}^2(u_{k;K}|v_{k;K},x_{{k};{K-1}}) \\& \times\prod_{n=1}^{K-1} f_{k:n}(x_{k:n}|x_{k:{n-1}})du_0dx_{1;1:{K-1}}du_{2;0}dx_{{2;1:{K-1}}}\ldots du_{M;K},
\end{align*} where $x_{1;0}=(v_{0},u_{0})$. This suggests the following sampling procedure:
\begin{compactenum}
    \item Begin by concatenating a sample from the conditional initial distribution $\mu^2(u_0|v_0)$ with $v_0$ to obtain $x_0=(v_0,u_0)=x_{1;0}$; each sample is weighted equally by $\mu^1(v_0)$.
    \item For each inter-observation interval $[t_{k-1},t_k]$, where a sample $x_{k;0}$ was available at time $t_{k-1}$, we extend the sample to time $t_{k;{K-1}}$, using the transition densities $\{f_{k;n}\}_{n=1}^{K-1}$, generating the ensemble $x_{k;1:K-1}$. The trajectory $x_{k;{0:K-1}}$ is then weighted by the marginal transition density $f_{k;K}^1(v_{k;K}|x_{{k};{K-1}})$.
    \item Finally, transit into the next time interval $[t_{k},t_{k+1}]$. At time $t_{k,K}$, we propose the partial latent state $u_{k;K}$ from $x_{k;{K-1}}$ by concatenating a sample from \mbox{$f_{k;K}^2(u_{k;K}|v_{k;K},x_{{k};{K-1}})$ with the known $v_{k;K}$ to obtain $x_{k+1;0}=x_{k;K}=(v_{k;K},u_{k;K})$}. 
\end{compactenum}
We have the following proposition, whose proof follows by factorising into the form of \eqref{generaltargetmeasure}, with the specified pairs as detailed below. 
\begin{proposition}
\label{propbridge_2}
The Feynman--Kac formulation, initialised at \\
\mbox{$M_0(u_0)= \mu^2(u_{0}|v_{0}),  G_0(u_0)= \mu^1(v_{0})$}, and defined, for $k\in[1:M]$:\\
\begin{minipage}{0.48\textwidth}     \begin{align*}
  \hspace{-5mm} M_{k;1}(x_{k;1}|u_{k-1;K})&=  f_{k;1}(x_{k;1}|v_{k-1;K},u_{k-1;K})\\
   \hspace{-5mm}          M_{k;n}(x_{k;n}|x_{k;{n-1}})&=f_{k;n}(x_{k;n}|x_{k;{n-1}}), 1<n<K\\
  \hspace{-5mm}           M_{k;K}(u_{k;K}|x_{k;{K-1}})&=f_{k;K}^2(u_{k;K}|v_{k;K},x_{{k};{K-1}})
        \end{align*}
\end{minipage}
\begin{minipage}{0.48\textwidth}
  \begin{align*}
           G_{k;n}(x_{k;n})&=1, n < K-1\\
         G_{k;K-1}(x_{k;K-1})&=f_{k;K}^1(v_{k;K}|x_{{k};{K-1}})\\
           G_{k;K}(u_{k;K})&=1
        \end{align*}
\end{minipage}
has normalising constant $f^{[K]}(v_{0:M})$, the $K$-step bridged partial pseudolikelihood.
\end{proposition}
\subsection{Splitting-based pseudolikelihoods}\label{l_assembly}
At last, {using splitting schemes,} we arrive at the explicit {approximating transition densities} required to make inference {via the cSMC algorithm}. We make clear that all the previous constructs are applicable for any numerical schemes fulfilling the required assumptions. For illustration, we dedicate this section to deriving identities, which are used to obtain the numerical results that are contained in this work. For an SDE within the class of \eqref{classofsde}, satisfying Assumption~\ref{assumption1}, we split \eqref{classofsde} into subequations \eqref{subeq1} and \eqref{subeq2},
and combine their solutions with the Lie--Trotter \eqref{LT} and Strang \eqref{S} compositions, leading to the explicit transition densities derived here.

\noindent{\bf Lie--Trotter-implied pseudolikelihoods}
Firstly, we introduce the identity {entering into Proposition~\ref{propbridge_1}} for computing the {Lie--Trotter-implied bridged full pseudolikelihood. Given the observations $\{x_k\}_{k=0}^{M}$, and inter-observation times $\{t_{k;n}\}_{n=0}^K$ being $\delta$ apart,  from \eqref{LT} we obtain $ f_{k;{n}}(x_{k;{n}}\mid x_{k;{n-1}})=\mathcal{N}(x_{k;{n}}; e^{A\delta}\Gamma_{\delta}(x_{k;{n-1}}), C(\delta))$, where $\mathcal{N}(x;a,b)$ denotes the density of a multivariate normal distribution with mean $a$ and covariance $b$ evaluated at $x$.} 
On the next level of complexity, we appeal to Proposition~\ref{prop1} and derive the Lie--Trotter-implied partial pseudolikelihood. For a process $x_k=(v_k,u_k)^\intercal$, where $v_k$ is the observed component with an observation interval $\Delta$, we make our notation more concise by defining
\begin{align}
\label{abbrev1}
    \mu_k= \begin{pmatrix}
    \mu_k^1\\\mu_k^2
\end{pmatrix}
:=&
e^{A\Delta}\Gamma_{\Delta}     
    \begin{pmatrix}
    v_{k-1}\\u_{k-1}
\end{pmatrix}\textrm{,\quad and\quad} &
\Sigma_k=
\begin{pmatrix}
        \Sigma_k^{1,1} & \Sigma_k^{1,2}\\
        \Sigma_k^{2,1} & \Sigma_k^{2,2}
\end{pmatrix}
:=& C(\Delta).
\end{align}
We have for each observation time $k$:
\begin{align*}
f_k^1(v_k|v_{k-1},u_{k-1})=&\mathcal{N}(v_k;\mu^1_k,\Sigma_k^{1,1})\\
f_k^2(u_k|v_k,v_{k-1},u_{k-1})=&\mathcal{N}\left(u_k;\mu_k^2+\Sigma_k^{2,1}(\Sigma_k^{2,2})^{-1}(v_k-\mu_k^1),\Sigma_k^{2,2}-\Sigma_k^{2,1}(\Sigma_k^{1,1})^{-1}\Sigma_k^{1,2}\right).
\end{align*}
Finally, we appeal to Proposition~\ref{propbridge_2} and derive the identities that give rise to the Lie--Trotter-implied bridged partial pseudolikelihood. We inherit the notation of \eqref{abbrev1} but with the finer time scale of $\{t_{k;n}\}_{n=0}^K$, that being $\delta$:
\begin{align*}
    f_{k;n}(x_{k;n}|x_{k;{n-1}})=&\mathcal{N}(x_{k;n}; e^{A\delta}\Gamma_{\delta}(x_{k;n-1}), C(\delta))\\
    f_{k;1}(x_{k;1}|v_{k-1;K},u_{k-1;K})=&\mathcal{N}(x_{k;1}; e^{A\delta}\Gamma_{\delta}((v_{k-1;K},u_{k-1;K})^\intercal), C(\delta))
\\
f_{k;K}^1(v_{k;K}|x_{{k};{K-1}})=&f_{k;K}^1(v_{k;K}|v_{k;K-1},u_{k;K-1})
    =\mathcal{N}(v_{k;K},\mu^1_{k;K-1},\Sigma_{k;K-1}^{1,1})\\
    f_{k;K}^2(u_{k;K}|v_{k;K},x_{{k};{K-1}})=&\mathcal{N}(u_{k;K};\mu_{{k};{K-1}}^2+\Sigma_{{k};{K-1}}^{2,1}(\Sigma_{{k};{K-1}}^{2,2})^{-1}(v_{t;M}-\mu_{{k};{K-1}}^1),\\&\phantom{\mathcal{N}(u_{k;K};}\Sigma_{{k};{K-1}}^{2,2}-\Sigma_{{k};{K-1}}^{2,1}(\Sigma_{{k};{K-1}}^{1,1})^{-1}\Sigma_{{k};{K-1}}^{1,2}).
\end{align*}
\noindent{\bf Strang-implied pseudolikelihoods}\label{non-invertible}
The Strang splitting scheme, {closely} related to the Lie--Trotter scheme, provides better one-step approximation but requires the well-definedness of the inverse of $\Gamma_\Delta${, the solution of the ODE \eqref{subeq2},} for the scheme to admit explicit transition densities, otherwise one should look for likelihood-free approaches, see e.g. \cite{SAMSON2025108095}. We require the explicitness of the Strang transition density, and hence impose the well-definedness of the inverse of $\Gamma_\Delta$:\begin{assumption}
\label{assump2}
The solution $\Gamma_\Delta$ of the ODE \eqref{subeq2} {admits an inverse function} $\Gamma_\Delta^{-1}${ which} exists everywhere on $\mathbb{R}^d$ {for $\Delta$ small enough}.
\end{assumption}
\noindent This assumption asserts that for a fine enough time-discretisation, the Strang scheme remains explicit, so that the half-step inverse of $X_{k}^{\textrm{S}}$ is a Lie--Trotter process, see \eqref{eqcond}.

Inferring the Strang-implied pseudolikelihoods is a two-step process: i. obtain the Lie--Trotter-implied pseudolikelihood on the inverse process $Z_k=\Gamma_\frac{\Delta}{2}^{-1}(X_k^{\textrm{S}})$; ii. straightforwardly address the change of variables, e.g. for every $f_k(x_k|x_{k-1})$ in \eqref{full}, 
we first estimate ${f}_{k}(z_k|z_{k-1})$, then compute ${f}_{k}(x_k|x_{k-1})={f}_{k}(z_k|z_{k-1})\left|\det{\left(D\Gamma_\frac{\Delta}{2}(\Gamma_\frac{\Delta}{2}^{-1}(x_k))\right)}\right|^{-1}$, where $D\Gamma$ denote the \textit{Jacobian matrix} of $\Gamma$ and $\det(\Gamma)$ its determinant. Note that the corresponding densities for the bridged case can be immediately recovered replacing $\Delta$ with $\delta$. In the partially observed case, we will require component-wise invertibility, in the sense that no unobserved coordinate enters the observed coordinate in the ODE subequation~\eqref{subeq2}, so that no unknown information is required to compute the transformation of variables.

The problem remains that, for datasets where the size of the time step is given, $\Gamma_\Delta$ might not be invertible for the {observation} time step $\Delta$. Again, bridging comes into play: in such applications, we will choose a bridge number $K$, such that, 
$\delta=\Delta/K$
{yields a discretisation fine enough}
for $\Gamma_\delta$ to be invertible. 

\subsection{Inference with cSMC-estimated pseudolikelihoods}\label{sec:bayesianinference}
Applying cSMC to the Feynman--Kac flows given in the previous Section gives estimates of the pseudolikelihoods that typically enjoy very low variances {compared to, e.g., the BPF}. These likelihood estimates can be used within any likelihood-based inferential framework which makes use of unbiased estimates of the likelihood. We describe here two particular approaches that we have found effective and which are employed in our numerical study.  
To obtain point estimates via a maximum (pseudo)likelihood approach, we will appeal to a noisy optimisation technique known as the simultaneous perturbation stochastic approximation (SPSA) method. From a Bayesian perspective, we will adopt a particle marginal Metropolis--Hastings (PMMH) approach {(see, \citealt{Pmcmc})} to approximate the posterior.

\noindent{\bf Point estimation: Simultaneous perturbation stochastic approximation}
The SPSA method, proposed by \citet{spsa}, is a recursive method which simultaneously updates all coordinates, leading to a significant improvement in efficiency compared with other stochastic approximation
methods that typically explore the impact of perturbations in one (or a subset of all) dimension(s) at a time. Here, we have adopted adaptive SPSA \citep{adaptive_spsa}, to further improve its efficiency. 
We refer to Supplementary Material \ref{SPSAsupp} for more details on this algorithm, its implementation in our setting and the proposed initialisation procedure.

\noindent{\bf Posterior exploration: Particle marginal Metropolis--Hastings } In the context of Bayesian computation, pseudomarginal methods \citep{Andrieu_2009} and particularly the particle MCMC methods proposed in \cite{Pmcmc}, have natural application in our framework.  Our work outlines a novel approach that leads to the unbiased estimates of various pseudolikelihoods with a low variance at a manageable computational cost. This is particularly important in the context of the {PMMH algorithm, as emphasised by \cite{Pittoptimalscaling}}.
Additionally, we will be able to achieve a better mixing performance with less dispersed likelihood estimates, 
{while also saving} computational resource relative to algorithms employing more naive estimates of the pseudolikelihood. We summarise our implementation of the PMMH in Algorithm~\ref{alg:pmmh} in Supplementary Material \ref{PMMHsup}. 
\section{Illustrative applications}\label{section5}
We now illustrate the performance of our proposed methodology on two examples, the one-dimensional cubic SDE, and the partially observed two-dimensional {hypoelliptic} stochastic FHN model. For the former, we show the explosion of the inference based on EuM compared to the considered splitting scheme. For the latter, we consider both frequentist and Bayesian inference from both simulated and real data under different observation regimes.

\noindent{\bf Implementation details.} All algorithms are implemented in the software environment {\bf R} \citep{R}, with codes available at \url{https://github.com/adamjohansen/csmc-diffusions}. We refer to Supplementary Material \ref{detailsupp} for further implementation details.
\subsection{The one-dimensional cubic SDE}\label{cubicsde_cs}
We first consider the cubic SDE as an introductory example: 
    $dX_t=-X_t^3dt+\sigma dW_t$, with $X_0=0$. 
As analysed in \cite{hutzen}, simulation-based methods that involve EuM approximation $X_k^{\textrm{EuM}}$ to the cubic SDE will suffer from divergent trajectories due to its not being a mean-square convergent scheme. We illustrate this using a Strang-generated cubic SDE  with $\sigma=100$, time step $\Delta=0.1$ and termination time $T_{\textrm{obs}}=100$. Specifically, it is generated using $\Delta_\textrm{sim}=10^{-5}$ and then subsampled with $\Delta_\textrm{obs}=0.1$. 
The data then have a length of $M_{obs}=1000$. As shown in Supplementary Material \ref{EuMSupp}, the EuM approach leads to 
explosions, and while using diffusion bridges partially mitigates this 
effect, we still observe it when the variance is 
large (e.g. in $40\%$ of cases when $\sigma=40$). Such explosions are entirely prevented when adopting the splitting schemes, as shown below.

\noindent{\bf Convergence of bridged pseudolikelihood}
As a preliminary study, consider point inference for the parameter $\sigma$. \cite{pedersen1995} provides consistency and asymptotic normality for the MLE based on the bridged {pseudo}likelihood \eqref{fullbridge} under Assumption~\ref{assumption_conv_likelihood}. We use the cubic example to explore this convergence numerically.

We generate the data using a moderate choice of diffusion parameter 
$\sigma=20$, the rest of the setup remains the same. We trace out the cSMC-inferred Lie--Trotter-implied bridged pseudolikelihood (on a log-scale) in Figure~\ref{fig:likelihood}, thus illustrating the impact of the bias-reduction method in action. The {ground} truth is simulated with 100 inter-observational steps. 
\begin{figure}[t]
    \centering    \includegraphics[width=0.6\linewidth]{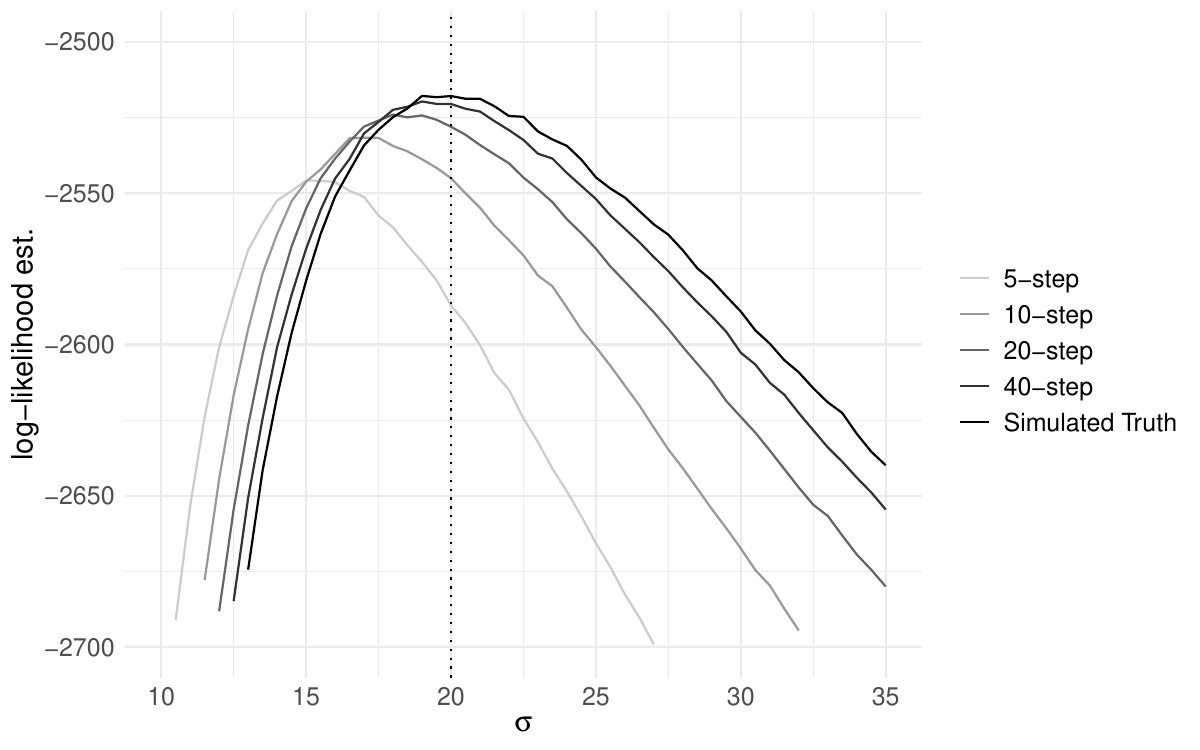}
    \caption{Inferred Lie--Trotter-implied bridged pseudolikelihoods (on a log-scale) {for the univariate cubic SDE \eqref{cubicsde_cs} for different number of timesteps in between observations}; fineness of grid: 0.5. The vertical dotted line shows the true parameter.}
    \label{fig:likelihood}
\end{figure}
Several observations follow. 
Firstly, the shape of the pseudolikelihood starts off peaking near $15$ and gradually  ``tilts'' upwards until approaching the simulated true likelihood which peaks at the truth. Secondly, the rate of the change in shape diminishes as we increase the number of bridging steps. This suggests that the bridged pseudolikelihood is indeed converging pointwise to the ``true'' likelihood. Although the cubic SDE does not admit an explicit solution, in this way, we are able to trace out what the ``actual'' likelihood should look like without employing more sophisticated time-discretisation tools.  
\subsection{The partially observed hypoelliptic FHN model}\label{fhn_cs}
The stochastic FHN model (\citealt{FitzHugh1955}) is a two-dimensional model for the single neuronal activity, with a super-linearly growing drift and additive noise:
\begin{equation}\label{FHN}
    d\begin{pmatrix}
        V_t\\U_t
    \end{pmatrix}=\begin{pmatrix}
        \frac{1}{\epsilon}(V_t-V_t^3-U_t)\\\gamma V_t-U_t+\beta
    \end{pmatrix}dt+\begin{pmatrix}
        \sigma_1&0\\
        0&\sigma_2
    \end{pmatrix}dW_t,
\end{equation}
{where $V_t$ denotes the membrane voltage, which is typically observed at discrete times, while $U_t$ is an unobserved variable modelling the ion channel kinetics.} When $\sigma_1=0$, the SDE is hypoelliptic, 
yielding a smooth transition density despite the rank-deficiency of the diffusion matrix, see, e.g., 
\citet[\textit{Condition 1}]{jrssb}. 
Hypoellipticity causes the components of the process to have variances which are not {of the same order in $\Delta$}, in the expression of the transition densities. The FHN is a good example for demonstrating our method, which does not require manual efforts to counter the effects of hypoellipticity as is the case for other methods, such as \citet{jrssb}. We forgo the possibility to use EuM, as it produces a degenerate Gaussian distribution as transition density, not being 1-step hypoelliptic{, differently from the considered splitting schemes}. 

\noindent{\bf Posterior exploration} 
In this section, we illustrate {our proposed} methodology for likelihood estimation in the context of Bayesian inference. We consider a high-frequency scenario where the number of observation is $M_{{\textrm{obs}}}=1000${, $T_{\textrm{obs}}=20$, i.e., the time step is} $\Delta_\textrm{obs}=0.02$. The true parameters are $(\epsilon,\gamma,\beta,\sigma_2)= (0.1,1.5,0.8,0.3)${ as in \cite{jrssb}}.  We impose a standard Gaussian prior on the parameters (in log-scale). 

Appealing to the baseline Feynman--Kac formulation for inferring the Strang {marginal pseudo}likelihood, with details in Supplementary Material \ref{fhn_spec}, we launch a PMMH chain of length $10^5$. We do this by using a BPF with 125 particles and a cSMC with 10 particles. The cSMC under our implementation is around twice as fast as the BPF and has negligible dispersion compared to it, as illustrated in Supplementary Material \ref{FHNsupp}.  We use the same unconditioned Gaussian random walk proposal to propose the new parameters (in log-scale) at each iteration of the PMMH.

We display the estimated {marginal} posterior densities in Figure~\ref{fig:fhn_mcmc}.  The BPF-inferred posterior densities have irregular shape, suggesting the MCMC chain is not well-mixed, stuck at several places for a significant amount of time. Instead, the cSMC-inferred posteriors are well-rounded and more informative about the actual (bell) shape of the posterior, while being significantly faster than the BPF.  
\begin{figure}[t]
    \centering
    \includegraphics[width=0.8\linewidth]{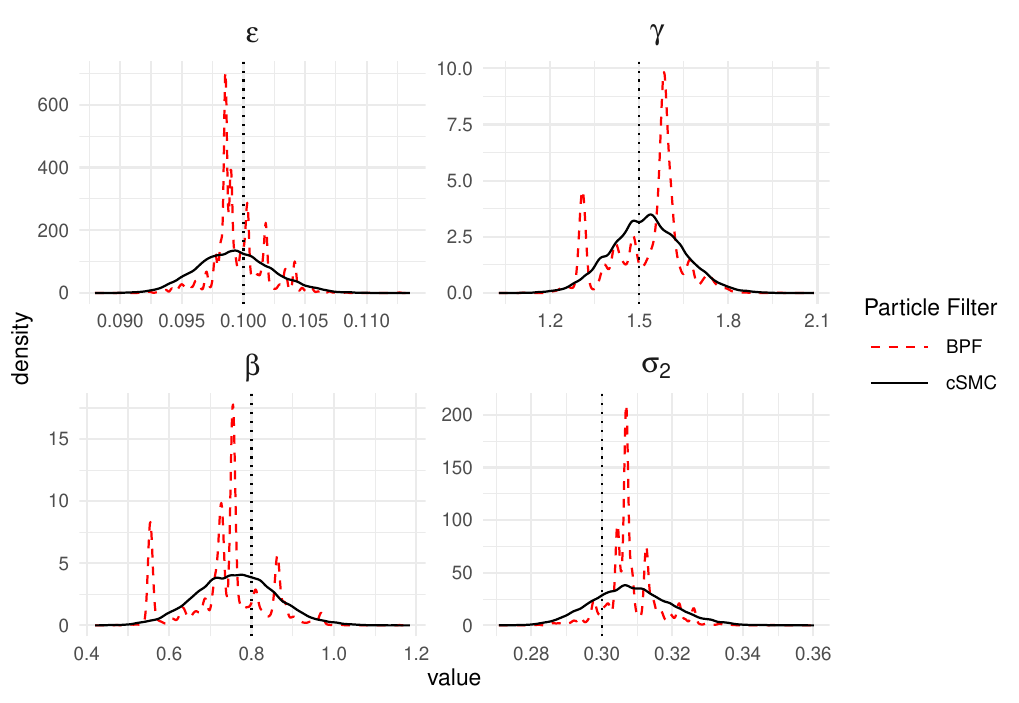}
    \caption{{Marginal} posterior density estimates {obtained via the BPF (red dashed) and the proposed cSMC approach (blue solid) for the partially observed hypoelliptic FHN model \eqref{FHN} with $\sigma_1=0$. The vertical dotted lines show the true parameters.}}
    \label{fig:fhn_mcmc}
\end{figure}

\noindent{\bf Point estimation} To illustrate the bias-reduction effect of diffusion bridging, we choose a low-frequency setting, where the number of observation is $M_{{\textrm{obs}}}=1000$, the time step is ${\Delta_{\textrm{obs}}}=0.05$ {for the same true parameters as before.}
We now compare several estimators: the unbridged Lie--Trotter and Strang MLEs, the 4-step bridged Lie--Trotter and Strang MLEs, the 8-step bridged Lie--Trotter and Strang MLEs. 

In Figure~\ref{fig:fhn_den}, we compare the density plots of the MLEs generated from 100 
independent trials using the pre-specified settings. The differences for the Lie--Trotter MLEs with 0,4,8 bridges {(grey solid, dashed, dotted lines)}, are quite clear: in particular, the shape of the distribution of the MLE \lq\lq leans\rq\rq over, changing dramatically while having its mean getting closer to the truth {as the number of bridges increases}. When we compare the  unbridged  Strang result to the unbridged Lie--Trotter, we see that the {former} has a strictly smaller bias in all four coordinates. Given that Strang has a higher weak order than Lie--Trotter, this is to be expected. Moreover, it is also very interesting to see that the estimated densities for Strang with 4 and 8 bridges {(black dashed, dotted lines)} are almost indistinguishable and peak remarkably close to the true parameters. Violin plots of the considered estimators are reported in Supplementary Material \ref{FHNsupp}.
\begin{figure}
    \centering
    \includegraphics[width=0.85\linewidth]{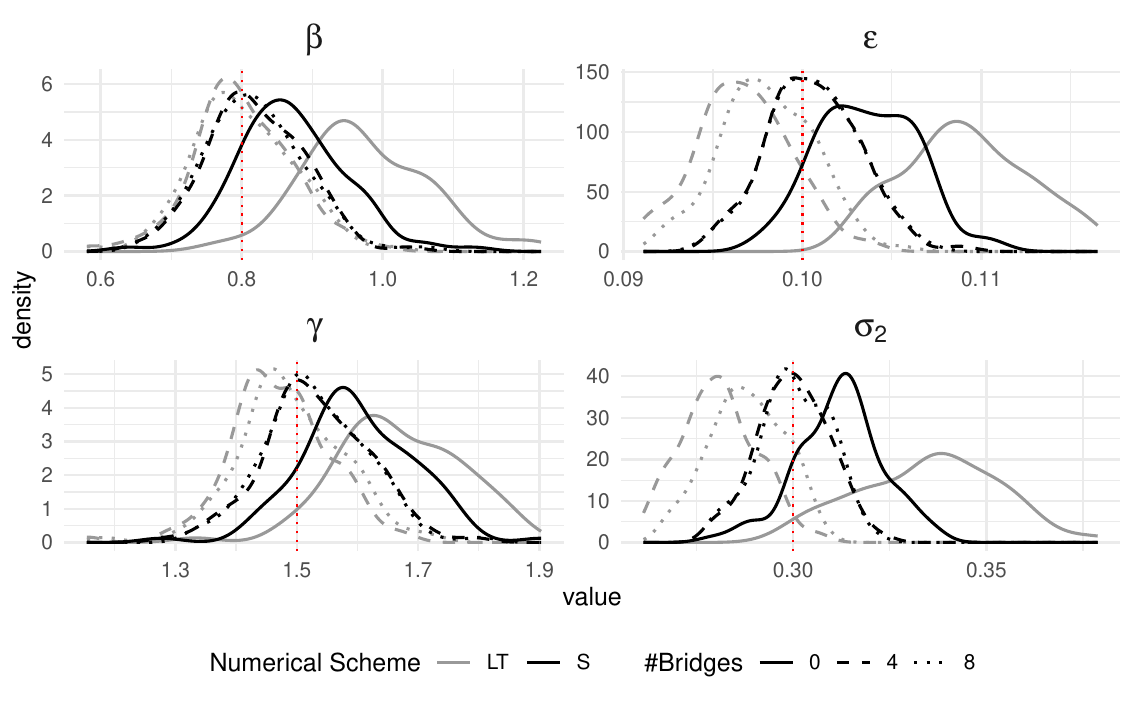}
    \caption{
   {Estimated density plots over 100 runs for Lie--Trotter MLE (grey lines) and Strang MLE (black lines) with 0 (solid lines), 4 (dashed lines), 8 (dotted lines) bridges for the partially observed hypoelliptic FHN model \eqref{FHN}  with $\sigma_1=0$}. The red vertical dotted lines show the true parameters.}
    \label{fig:fhn_den}
\end{figure}
\subsection{Real world application: Modelling neural activity}\label{application_rats}
{We now apply our methodology on membrane voltages 
recorded from the dorsal rootlet of rats. The data, available as open source and documented by \cite{ratdata},  has been recently used by \cite{SAMSON2025108095} to successfully fit a hypoelliptic FHN model, with marginal posterior densities estimated via Approximate Bayesian Computation (ABC).}
We will {focus our} attention to the dataset 1609 within the Cutaneous Stimulation category, {with data recorded with step size $\Delta_\textrm{obs}=0.02$ up to time $T_\textrm{obs}=250$, so $n_\textrm{obs}=12500$.}

{As the observed dataset is centred close to zero, with empirical mean 
0.00263, following \cite{SAMSON2025108095}, we will assume} that the data is collected as $V_{data}=V_{raw}-\mu_{raw}$, 
where $V_{data}$ is the centred data process, and $V_{raw}$  is 
the observed component for the FHN which has the censored mean $\mu_{raw}$. 
For a set of parameters $\theta$, we aim to retrieve the observed component of the FHN as $\hat{V}_{raw}^{\theta}=V_{data}+\mu^{\theta}_{raw}$, where $\mu^{\theta}_{raw}$ is an estimate of the censored mean $\mu_{raw}$. Due to the FHN being geometrically ergodic, we {will estimate $\mu^\theta_{raw}$ as the invariant mean of the observed $V$ component via Monte Carlo estimate, obtained simulating a Strang-trajectory under $\theta$ with $10^5$ observations and time step $\Delta_\text{sim}=0.02$.}

We will conduct inference in the unbridged case, following procedurally what was done to the simulated data, estimating the Strang-implied pseudolikelihood using cSMC, and maximising it using (adaptive) SPSA. The convergence of the estimator is illustrated in Supplementary Material \ref{FHNsupp}. The final point estimate was ${\hat\theta=}(0.283,\allowbreak 72.7,\allowbreak 87.3,\allowbreak 0,\allowbreak 2.46)$, producing simulated trajectories that resemble the data quite well in terms of the overall amplitudes and frequency (see Figure~\ref{sim_traj} in Supplementary Material \ref{FHNsupp}). 
{While our estimates and those of \cite{SAMSON2025108095} are comparable on simulated data, validating our proposed methodology, we observe a notable discrepancy on real data, with their ABC posterior mean being $(0.033,\allowbreak 6.701,\allowbreak 3.995,\allowbreak 0.133)$.}
One possible cause is model misspecification, suggesting that the FHN-probabilistic model misfits the data. A weak identifiability issue might also exist due to the missing mean. Finally, it is important to highlight that the two methods target different objectives, hence the discrepancy in the estimates, as \citet{SAMSON2025108095} inform the ABC well of the data's oscillatory behaviour with summary statistics, whereas our method directly maximises the inferred pseudolikelihood with no tailored information about the shape of the data. 
\section{Conclusions}
cSMC is a powerful tool when it comes to intractable likelihood problems when the latent process is sufficiently tractable. Combining it with good time-discretisation schemes{, here splitting methods,} has led to a framework that tackles inference for SDE from a likelihood-based perspective that was considered to be very challenging and inefficient.  Not only we do provide a compact package to infer about the parameters across a wide range of observation settings, but we also propose a generic way to reduce the systematic bias arose from time-discretisation.  Without the help of cSMC, this was seen as 
{both} practically infeasible, due to the extensive computational costs{, and statistically unreliable, as performing Monte Carlo sampling from high dimensional integrals would have led to estimates with large variance. While we focused on numerical splitting schemes, as they have been proved to be mean-squared convergent and property preserving for the considered class of semi-linear SDEs with additive noise, it is important to stress that the proposed methodology can be directly applied with any suitable numerical scheme, as long as the numerical transition density is an invertible function of a Gaussian.} 
Future perspective of this work may involve relaxing this Gaussian-dependent condition onto the numerical schemes. This way we might be able to incorporate even more advanced numerical schemes, further expanding the range of problems and SDEs that this method can deal with.  
\section*{Acknowledgements}
We are grateful to the Scientific Computing Research Technology Platform (SCRTP) at the University of Warwick for the provision of computational resources. SH acknowledges funding from the China Scholarship Council (CSC), RGE acknowledges the financial support of the United Kingdom Engineering and Physical Sciences Research Council (EPSRC) grant EP/W006790/1, while MT the EPSRC grant EP/X020207/1. AMJ acknowledges financial support by EPSRC grant EP/R034710/1 and by United Kingdom Research and Innovation (UKRI) via grant number EP/Y014650/1, as part of the ERC Synergy project OCEAN. For the purpose of open access, the authors have applied a Creative Commons Attribution (CC BY) licence to any Accepted Manuscript version arising from this submission. 

\noindent \textbf{Data access:} The considered real data is available on Mendeley Data at \url{https://data.mendeley.com/datasets/ybhwtngzmm/1}, while the simulated datasets can be obtained using the provided {\bf R} codes available at \url{https://github.com/adamjohansen/csmc-diffusions}.

\appendix
\section{Specification for the cubic SDE}
With the usual convention that $X_k$ denotes the process $X$ at time $t_k$, the EuM discretisation \eqref{EuM} of the cubic SDE \eqref{cubicsde_cs} is given by:
\[
X^{\textrm{EuM}}_k=X^{\textrm{EuM}}_{k-1}-\Delta(X^{\textrm{EuM}}_{k-1})^3 +\xi_{\Delta,k}^{\textrm{EuM}}, \qquad\xi^{\textrm{EuM}}_{\Delta,k}\overset{\text{iid}}{\sim}N(0,\Delta\sigma^2),
\]
so $f_k^{\textrm{EuM}}(x_k|x_{k-1})=\mathcal{N}(x_k;x_{k-1}-\Delta x_{k-1}^3,\Delta\sigma^2)$, with $\Delta, x_{k-1}, x_k$ replaced by $\delta, x_{k;n-1}, x_{k;n}$ when considering the bridged approximated transition density $f_{k;n}(x_{k;n}|x_{k;n-1})$.

In \cite{buckwar2022splitting}, the cubic SDE is decomposed as \eqref{subeq1}-\eqref{subeq2}, with $A=-1$ and $\gamma(X^{[2]}_t)=X^{[2]}_t-(X^{[2]}_t)^3$, leading to a LT  \eqref{LT} and Strang \eqref{S} schemes, with $\xi_{\Delta,k}\stackrel{\text{iid}}{\sim}N(0,\sigma^2(1-e^{-2\Delta}))$ and $\Gamma_\Delta(x)={x}/\sqrt{e^{-2\Delta}+x^2(1-e^{-2\Delta})}$, and thus $f_k^{\text{LT}}(x_k|x_{k-1})=\mathcal{N}\left(x_{k};e^{-\Delta}\Gamma_\Delta(x_{k-1}),\sigma^2(1-e^{-2\Delta}\right)$, with $f_{k;n}(x_{k;n}|x_{k;n-1})$ immediately obtained from it, replacing $\Delta, x_{k-1}, x_k$ with $\delta, x_{k;n-1}, x_{k;n}$, respectively. 
\section{Specification for the partially observed FHN}\label{fhn_spec}
{Following \cite{buckwar2022splitting}, the FHN SDE \eqref{FHN} is decomposed as \eqref{subeq1}-\eqref{subeq2}, with: 
\begin{eqnarray}
   A= \begin{pmatrix}
     0&-\frac{1}{\epsilon}\\
     \gamma &-1
    \end{pmatrix},\qquad
    \gamma(X_t)=\begin{pmatrix}
        \frac{1}{\epsilon}(V_t-V_t^3)\\
        \beta
    \end{pmatrix},
\end{eqnarray}
leading to a LT  \eqref{LT} and Strang \eqref{S} splitting schemes. Denoting $\kappa=4\gamma/\epsilon-1$, for $\kappa>0$, the quantities entering into these schemes are:}
\begin{eqnarray*}
    e^{A\Delta} &=& e^{-\frac{\Delta}{2}} 
\begin{pmatrix}
\cos\left(\frac{1}{2}\sqrt{\kappa \Delta}\right) + \frac{1}{\sqrt{\kappa}}\sin\left(\frac{1}{2}\sqrt{\kappa \Delta}\right) & 
-\frac{2}{\epsilon \sqrt{\kappa}}\sin\left(\frac{1}{2}\sqrt{\kappa \Delta}\right)\\
\frac{2\gamma}{\sqrt{\kappa}}\sin\left(\frac{1}{2}\sqrt{\kappa \Delta}\right) & 
\cos\left(\frac{1}{2}\sqrt{\kappa \Delta}\right) - \frac{1}{\sqrt{\kappa}}\sin\left(\frac{1}{2}\sqrt{\kappa \Delta}\right)
\end{pmatrix},\\
    C(\Delta) &=& 
\begin{pmatrix}
c_{11}(\Delta) & c_{12}(\Delta) \\
c_{12}(\Delta) & c_{22}(\Delta)
\end{pmatrix},
\\
c_{11}(\Delta) &=& \frac{\sigma^2 e^{-\Delta}}{2\epsilon \gamma \kappa} \left(-\frac{4\gamma}{\epsilon} + \kappa e^{\Delta} + \cos(\sqrt{\kappa \Delta}) - \sqrt{\kappa}\sin(\sqrt{\kappa \Delta})\right), \\
c_{12}(\Delta) &=& \frac{\sigma^2 e^{-\Delta}}{\kappa \epsilon} \left(\cos(\sqrt{\kappa \Delta}) - 1\right), \\
c_{22}(\Delta) &=& \frac{\sigma^2 e^{-\Delta}}{2\kappa} \left(\cos(\sqrt{\kappa \Delta}) + \sqrt{\kappa}\sin(\sqrt{\kappa \Delta}) - \frac{4\gamma}{\epsilon} + \kappa e^{\Delta}\right),
\end{eqnarray*}
{with $\cos\left(\sqrt{\kappa}t/2\right) = \cosh\left(\sqrt{-\kappa}t/2\right)$ and $  
\sin\left(\sqrt{\kappa}t/2\right)/\sqrt{\kappa} = \sinh\left(\sqrt{-\kappa}t/2\right)/\sqrt{-\kappa}$ if $\kappa<0$. Moreover, $\Gamma_\Delta, \Gamma_\Delta^{-1}$ and $D\Gamma_\Delta$ are given by:} 
\begin{equation*}
    \Gamma_{\Delta}\begin{pmatrix}
        x\\y
    \end{pmatrix} = 
\begin{pmatrix}
\frac{x}{\sqrt{e^{-\frac{2\Delta}{\epsilon}} + x^2 \left(1 - e^{-\frac{2\Delta}{\epsilon}}\right)}} \\
\beta \Delta + y
\end{pmatrix},\qquad    \Gamma_{\Delta}^{-1}\begin{pmatrix}
        x\\y
    \end{pmatrix}=\begin{pmatrix}
        \textrm{sign}(x)\sqrt{\frac{e^{-\frac{2\Delta}{\epsilon}}x^2}{1-(1-e^{-\frac{2\Delta}{\epsilon}})x^2}}\\y-\beta\Delta
    \end{pmatrix},
\end{equation*}
\begin{equation*}
    D\Gamma_{\Delta}\begin{pmatrix}
        x\\y
    \end{pmatrix}=\begin{pmatrix}
        \left(e^{-\frac{2\Delta}{\epsilon}} + v^2 (1 - e^{-\frac{2\Delta}{\epsilon}})\right)^{-\frac{1}{2}} - \frac{v^2 (1 - e^{-\frac{2\Delta}{\epsilon}})}{\left(e^{-\frac{2\Delta}{\epsilon}} + v^2 (1 - e^{-\frac{2\Delta}{\epsilon}})\right)^{3/2}}\\1
    \end{pmatrix}.
\end{equation*}
\section{Proof of Proposition 1}\label{Proof}
\begin{proof}[Proof of Proposition \ref{prop1}.]
Rearrangement and Fubini's Theorem yield:
\begin{align*}
    &f(v_{0:M}) =\int_{\mathcal{U}^{M+1}}\mu^1(v_0) \mu^2(u_0|v_0) \prod_{k=1}^{M} f_k^1(v_{k}|v_{k-1},u_{k-1})f_k^2(u_k|v_k,u_{k-1},v_{k-1})du_{0:{M}}\\
    =&\int_{\mathcal{U}}\int_{\mathcal{U}^{M}}\left[ \mu^1(v_0) \mu^2(u_0|v_0) f_1^1(v_1|v_{0},u_{0}) \prod_{k=1}^{M-1} f_{k+1}^1(v_{k+1}|v_{k},u_{k})f_k^2(u_k|v_k,u_{k-1},v_{k-1}) \right.\\
    &\left.\times f_{M}^2(u_M|v_M,u_{M-1},v_{M-1})\right] du_{0:{M-1}} d u_{M}\\
     =& \int_{\mathcal{U}^{M}}\left[\mu^1(v_0) \mu^2(u_0|v_0) f_1^1(v_1|v_{0},u_{0}) \prod_{k=1}^{M-1} f_{k+1}^1(v_{k+1}|v_{k},u_{k})f_k^2(u_k|v_k,u_{k-1},v_{k-1})\right. \\
    &\left.\times \int_{\mathcal{U}} f_{M}^2(u_M|v_M,u_{M-1},v_{M-1})\right]d u_{M}du_{0:{M-1}} \\
    =&\int_{\mathcal{U}^{M}}\mu^1(v_0) f_1^1(v_1|v_{0},u_{0}) \mu^2(u_0|v_0) \prod_{k=1}^{M-1} f_{k+1}^1(v_{k+1}|v_{k},u_{k})f_k^2(u_k|v_k,u_{k-1},v_{k-1})du_{0:{M-1}}\\
    =&\int_{\mathcal{U}^{M}}G_0(u_0)M_0(u_0)\prod_{k=1}^{M-1}G_k(u_k)M_k(u_{k-1},u_k)du_{0:M-1},
\end{align*}
with $G_0, M_0, G_k, M_k$ defined as in \eqref{prop1eq1}-\eqref{prop1eq2}.
\end{proof}

\section{SPSA algorithm}\label{SPSAsupp}
The SPSA method, proposed by \cite{spsa}, is a numerical optimisation method suitable for optimising functions from which only noisy evaluations are possible. It is a recursive method which simultaneously updates all coordinates, requiring only a small number of noisy objective function evaluations per iteration. In scenarios in which the dominant contribution to computational cost is the evaluation of the objective function, {here, the pseudolikelihood,} SPSA will lead to a significant improvement in efficiency compared with other stochastic approximation methods that typically explore the impact of perturbations in one (or a subset of all) dimension(s) at a time. We summarise our implementation of the SPSA in Algorithm~\ref{alg::spsa}, where $\hat{L}^{cSMC}(\theta)$ is the cSMC-estimated likelihood of a given $\theta$; in the numerical optimisation context, this is often referred to as a \textit{measurement}. We make clear that gains $a_k$ and directions $c_k$ in the algorithm are sequences of scalars, with scaling matrices $s_k$ in $\mathbb{R}^{p\times p}$ (with further implementation details reported in Supplementary Material \ref{detailsupp}). 
\begin{algorithm}[t]
\caption{Simultaneous Perturbation Stochastic Approximation {(SPSA)}}\label{alg::spsa}
\textbf{Input}: gains $(a_k)$, directions $(c_k)$, scaling matrices $(s_k)$, initial $\theta_0$, \# steps $\tau_{end}$.
\begin{algorithmic}
\State Set $k=1$, compute $p=\textrm{dim}(\theta_0)$ the dimension of the parameter vector  
\While{$k<\tau_{end}$}
\State sample $\Delta_k$ uniformly over $\{-1,+1\}^p$ independently of $\Delta_{1:k-1}$
\State compute $y_k^{\pm} = \hat{L}^{cSMC}(\theta_{k-1}\pm c_k\Delta_k)$ and $g_k=\frac{y_k^+-y_k^-}{2c_k\Delta_k}$
\State update $\theta_{k}=\theta_{k-1}+a_ks_kg_k$; $k \leftarrow k+1$
\EndWhile
\end{algorithmic}
\textbf{Output}: The maximum likelihood estimate $\theta_{\tau_{end}}$
\end{algorithm}
We have adapted a basic notion of parameter scaling: the SPSA searches each direction by the same amount at each iteration and updates along each direction a scaled amount determined by $s_k$. For any numerical demonstration included in this work, we have adopted the adaptive scaling machine $s_k$ described by \citet[equation 2.1]{adaptive_spsa} that takes an extra {measurement} to approximate the second-order stochastic gradient. Our implementation consequently uses 5 {measurements} for each proposal. 

We now outline our initialisation procedure. The initialisation is typically robust for the unbridged case, as the variance of the cSMC estimate is quite small compared with the bridged case as a manifestation of the curse of dimensionality. The variance of the cSMC estimate typically grows with the number of latent bridging steps. As a result, initialisation for the unbridged case requires the minimal manual effort. In the bridged case, we propose to initialise the SPSA by a run of the unbridged case, using only half of the particles, and start the SPSA from the resulting set of parameters.  For a deterministic stopping rule $\tau_{end}$, this will marginally increase the cost of running a $K$-step bridged inference by a factor of $(1+1/2K)$.

\section{Particle marginal Metropolis-Hastings algorithm}\label{PMMHsup}
We summarise our implementation of the PMMH in Algorithm~\ref{alg:pmmh}.
\begin{algorithm}[t]
\caption{PMMH Algorithm targeting the posterior distribution}\label{alg:pmmh}
\textbf{Input}: initial state $\theta_{0}$, proposal $q({\theta^*{\mid}\theta})$, prior $\pi(\theta)$, \# iterations $I$
\begin{algorithmic}
\State compute $L_0=\hat{L}^{cSMC}(\theta_0)$
\For{$i=1,\ldots, I$}
\State sample $\theta^*$ from $q({\theta^*{\mid}\theta_{i-1}})$ and estimate $L^*=\hat{L}^{cSMC}(\theta^*)$
\State with probability $\min(1,\frac{L^*\pi(\theta^*)q({\theta^*{\mid}\theta_{i-1}})}{L_{i-1}\pi(\theta_{i-1})q({\theta_{i-1}{\mid}\theta^*})})$ set $(\theta_i,L_i)=(\theta^*,L^*$);
\State otherwise, set $(\theta_i,L_i)=(\theta_{i-1},L_{i-1})$.
\EndFor
\end{algorithmic}
\textbf{Output}: $(\theta_{i})_{i=1}^N$
\end{algorithm}
\section{Implementation details}\label{detailsupp}
We implemented our method in {\bf R} \citep{R}. All components of the algorithms have been coded from scratch including Algorithms~\ref{alg:pf_simple} and \ref{alg::adp} and the cSMC flow detailed by Figure~\ref{fig:csmcflow} in the main manuscript. We used the \texttt{mvnfast} package \citep{mvnfast} for fast sampling from multivariate Gaussians. Occasionally, Algorithm~\ref{alg::adp} will produce a non-positive definite $A_k${, leading to a  corresponding degenerate Gaussian variable}. 
We adopt a postponing strategy and replace such variables with one that is component-wise independent and has flat densities (big variances) in all directions. 
{Hence,} the density is an analogy of a constant policy (hence a BPF), with the hope that it will be refined at the next iteration of the cSMC, as in fact observed in all explored examples. We implemented the adaptive SPSA detailed by \cite{adaptive_spsa} {as described in Supplementary Material \ref{SPSAsupp}}, respecting guidance of choosing search and gain sequences $c_k$, $a_k$ from \cite{spsa} and scaling matrices $s_k$ from \cite{adaptive_spsa}. For each measurement of SPSA, we found that 20 particles per iteration of the cSMC are enough to deliver stable estimates. {In particular, in the bridged case, we fixed the number of iterations of SPSA to be 200, and observed that an average run of cSMC takes 2 or 3 iterations to converge.} For the bridged cases, we found that for our length of observation $M_{obs}=1000$, the naive BPF evaluated at a current set of parameters is not a good starting point of the cSMC and led severe exhaustion of particles. We therefore appealed to the coarse initialisation trick detailed by \citet[Chapter 11]{pierthesis}, while also storing the policies generated from the last measurement of the SPSA as the starting policies for the next measurement (the latter only to save computing time). We recall no manual effort when working with the PMMH other than a slight pre-conditioning of the symmetric Gaussian proposal $q(\theta^*|\theta_{i-1})$ with guidance provided by \citet[Chapter 16]{introtosmc}. We report that all plots have been generated with \texttt{ggplot2} \citep{ggplot} in {\bf R}; for plots showing kernel density estimates, the \emph{rule of thumb} suggested by \citet[p. 48]{silverman1986} was employed to choose the bandwidth.

All case studies have been run on multiple core SCRTP Taskfarm at the University of Warwick (Intel Xeon Gold 6248R,
3.0 GHz, 24-core processors). To show the portability of the method, the reported runtimes have been obtained on an Apple M3 pro (P cores 4.05 GHz, E cores 2.75 GHz, 12 cores).
\begin{figure}[t]
    \centering    \includegraphics[width=.9\textwidth]{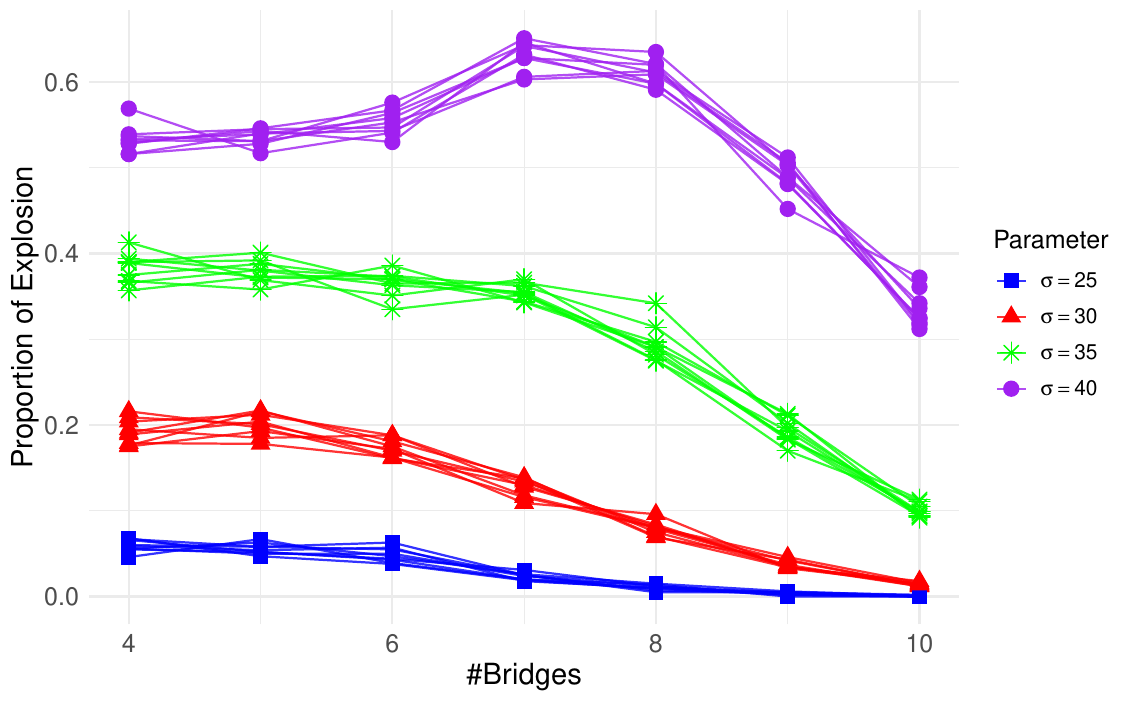}
    \caption{Illustration of the EuM blowing up {for the univariate cubic SDE \eqref{cubicsde_cs}}: 
Proportion of observation intervals out of 1000 where explosion {(i.e. values above $10^5$)} happened as a function of the number of latent bridging steps for different values of $\sigma$.}
    \label{fig:eumexplode}
\end{figure}
\section{Explosion of EuM-based cSMC}\label{EuMSupp}
Using Proposition~\ref{propbridge_1}, we derive the Feynman--Kac flow $\{G_{k;n},M_{k;n}\}_{n=1}^{K-1}$ to infer the $K$-step bridged EuM transition densities at time points $\{t_k\}_{k=0}^{M_{obs}}$. The cSMC in its first iteration, which is a standard BPF, attempts to simulate a trajectory, consisting of $K-1$ steps starting from $x_{k;0}=x_{k-1}$, using the EuM dynamic. Such trajectories will be corrected at the end of the bridge using the only non-constant potential function $G_{K-1}$.
{As expected}, using the EuM scheme, we indeed observed explosions between observations, happening at an increasing chance as the diffusion parameter $\sigma$ grows{, and decreasing as the number of latent bridging steps increases (and thus the time step $\delta$ decreases). This can be observed in Figure \ref{fig:eumexplode}, where we report the proportion of intervals out of 1000
(for ten different trials for each chosen $\sigma$), in which  several particles \lq\lq explode\rq\rq (i.e., go above $10^5$) after being propagated through the bridge in the first iteration of the cSMC, while others \lq\lq properly behave\rq\rq, staying near the $[-5,5]$ range.} As a result, the next-in-line step of the cSMC, which is the linear regression policy learning, will perform poorly with outliers, and demand manual efforts to a large extent. 

\section{Additional results for the FHN model}\label{FHNsupp}
\begin{figure}
    \centering
    \includegraphics[width=0.9\linewidth]{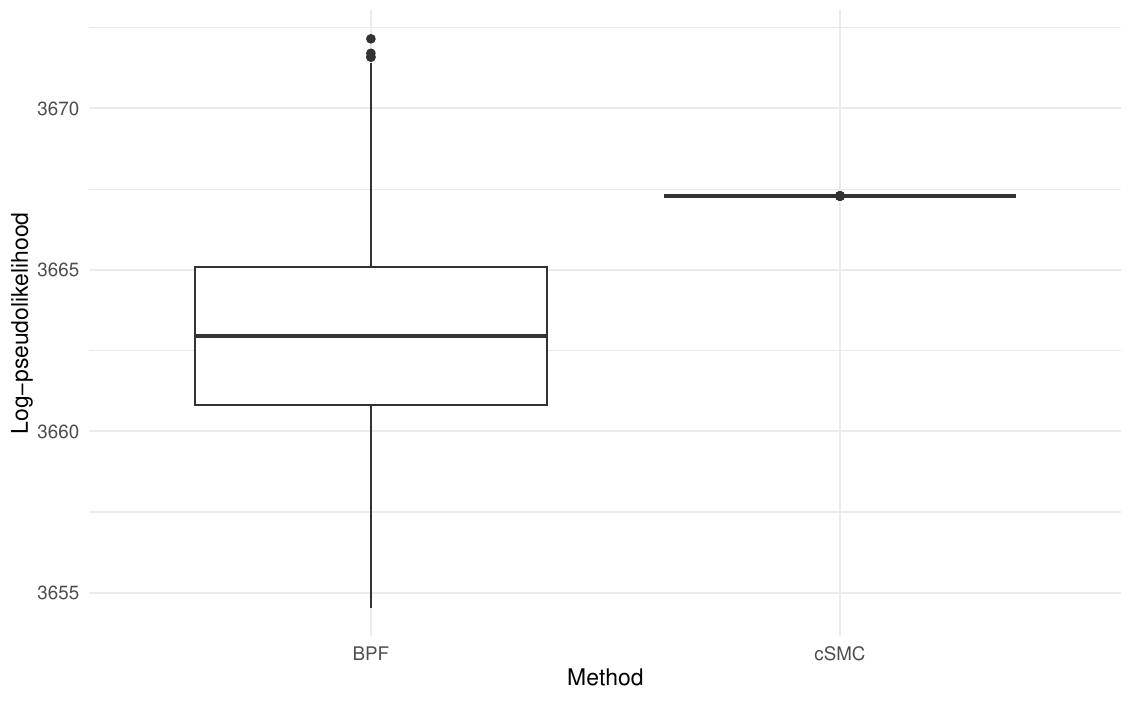}
    \caption{Boxplot of partial pseudolikelihood (on a log-scale) evaluated at the posterior mode from 1000 independent runs of both methods.}
    \label{fig:bpfvcsmc}
\end{figure}
\noindent {\bf Simulated data --- Posterior exploration} Here, we start by providing additional results for the posterior exploration of the partially observed FHN model described in Section \ref{fhn_cs} of the main manuscript. As a demonstration of the {performance of} cSMC, a comparison {between the partial pseudolikelihood estimated via} 
BPF and cSMC evaluated at the posterior mode $(0.0993, 1.53 ,0.763 ,0.308)$ is reported in Figure~\ref{fig:bpfvcsmc}. In this rather extreme case, the cSMC estimate behaves almost like a constant, with visibly no dispersion compared with the BPF estimates.

\noindent {\bf Simulated data --- Point estimation} {In Figure ~\ref{fig:fhn_violin}, we report the violin plots (with embedded boxplots) of the six Lie--Trotter and Strang MLEs with 0, 4 and 8 bridges obtained for the partially observed FHN model, complementing Figure \ref{fig:fhn_mcmc} therein.} 
Although an overall strict superiority is shown by the Strang scheme, we see a competing performance between the bridged Lie--Trotter MLE and the bridged Strang MLE on the $\gamma$ and $\beta$ coordinates. 

Finally, we comment on the runtime to reproduce Figure \ref{fig:fhn_den} in the main manuscript. With an entirely serial implementation in {\bf R}, the per-measurement time for the 8-step bridged trials is roughly 20 seconds, while for the 4-step bridged case it is 10 seconds. We let the SPSA run until convergence, taking at most 200 iterations. With 5 measurement per iteration of the adaptive SPSA, each independent trial takes at most 4.5 hours in the 8-step bridged case. For real data with a length of 12500, we report that we let SPSA to run 1500 iterations until convergence, with each iteration costing roughly 50 seconds.

\begin{figure}[t]
    \centering
    \includegraphics[width=0.85\linewidth]{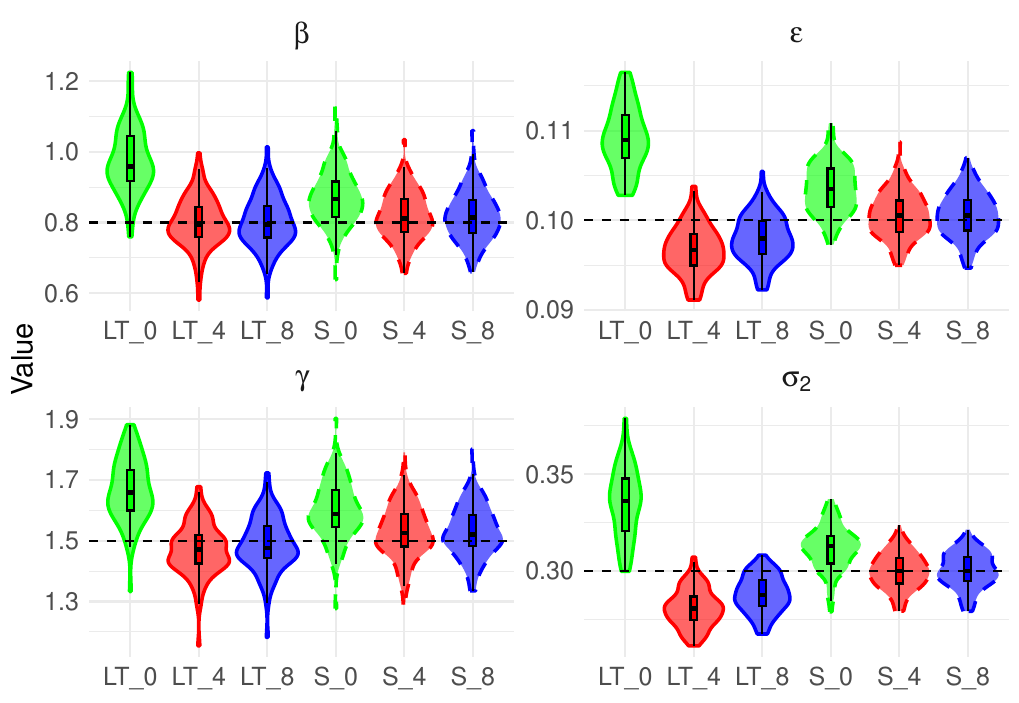}
    \caption{Violin plots {(with embedded boxplots) for the Lie--Trotter (LT) and Strang (S) MLEs with 0, 4, and 8 bridges. The true parameters are represented with horizontal lines.}}
     \label{fig:fhn_violin}
\end{figure}

\noindent{\bf Real data} Here, we complement {the details} and results on inference for the real dataset reported in Section \ref{application_rats} of the main manuscript. {As we have no prior information about the parameters, we did a few very short exploratory SPSA runs and set $(0.2,50,50,0.1)$
{as starting point.} }   
In Figure \ref{fig:1609spsa}, we report the convergence trajectory for a single run of the SPSA algorithm used to maximise the Strang-implied partially observed pseudolikelihood estimated via cSMC, {while in Figure \ref{sim_traj} we report the true and simulated data under the inferred parameters.}
\begin{figure}
    \centering    \includegraphics[width=0.8\textwidth]{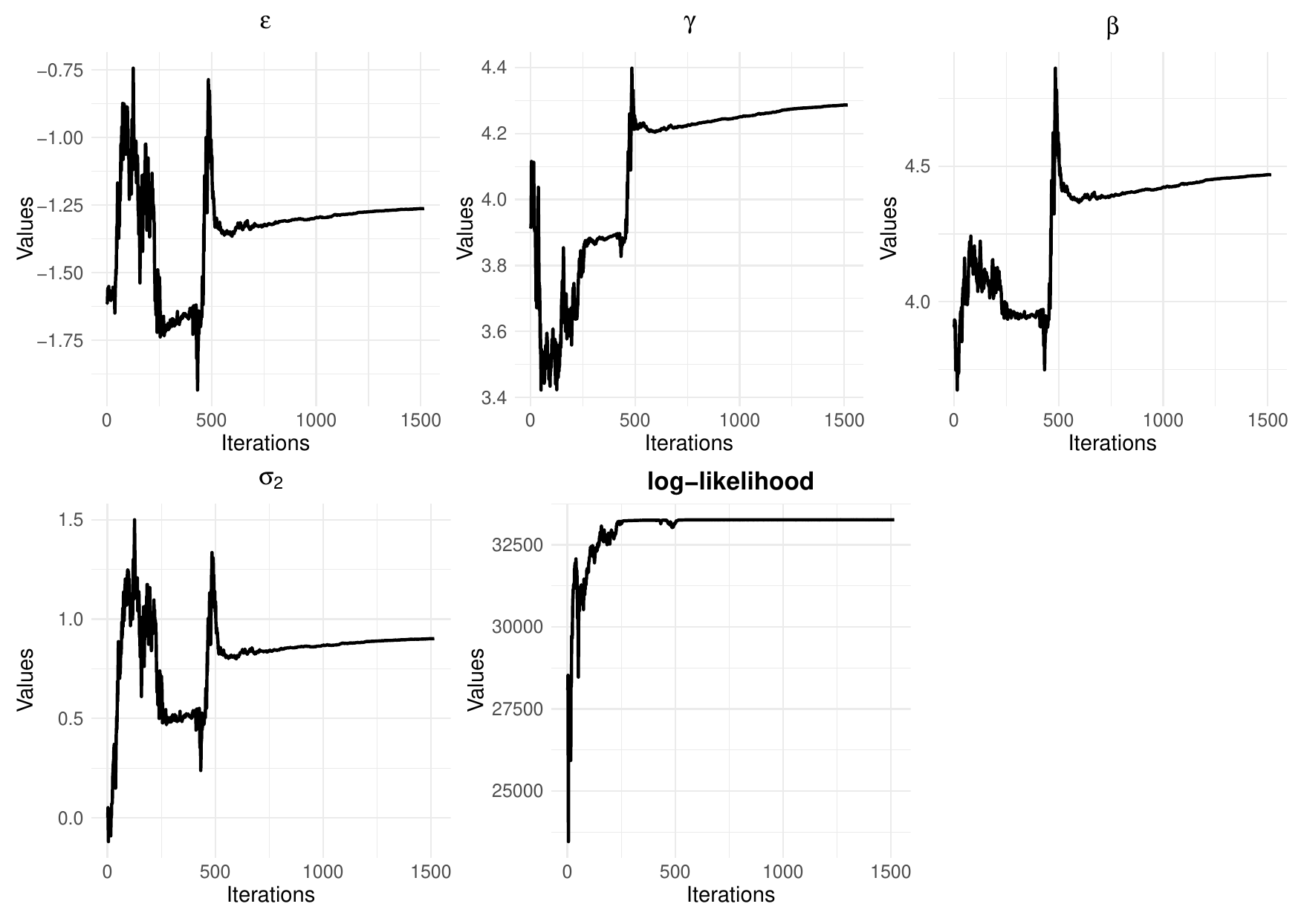}
    \caption{Convergence trajectory for a single run of SPSA. Parameters are on log-scale. }
    \label{fig:1609spsa}
\end{figure}
\begin{figure}
  \centering
\includegraphics[width=0.85\textwidth]{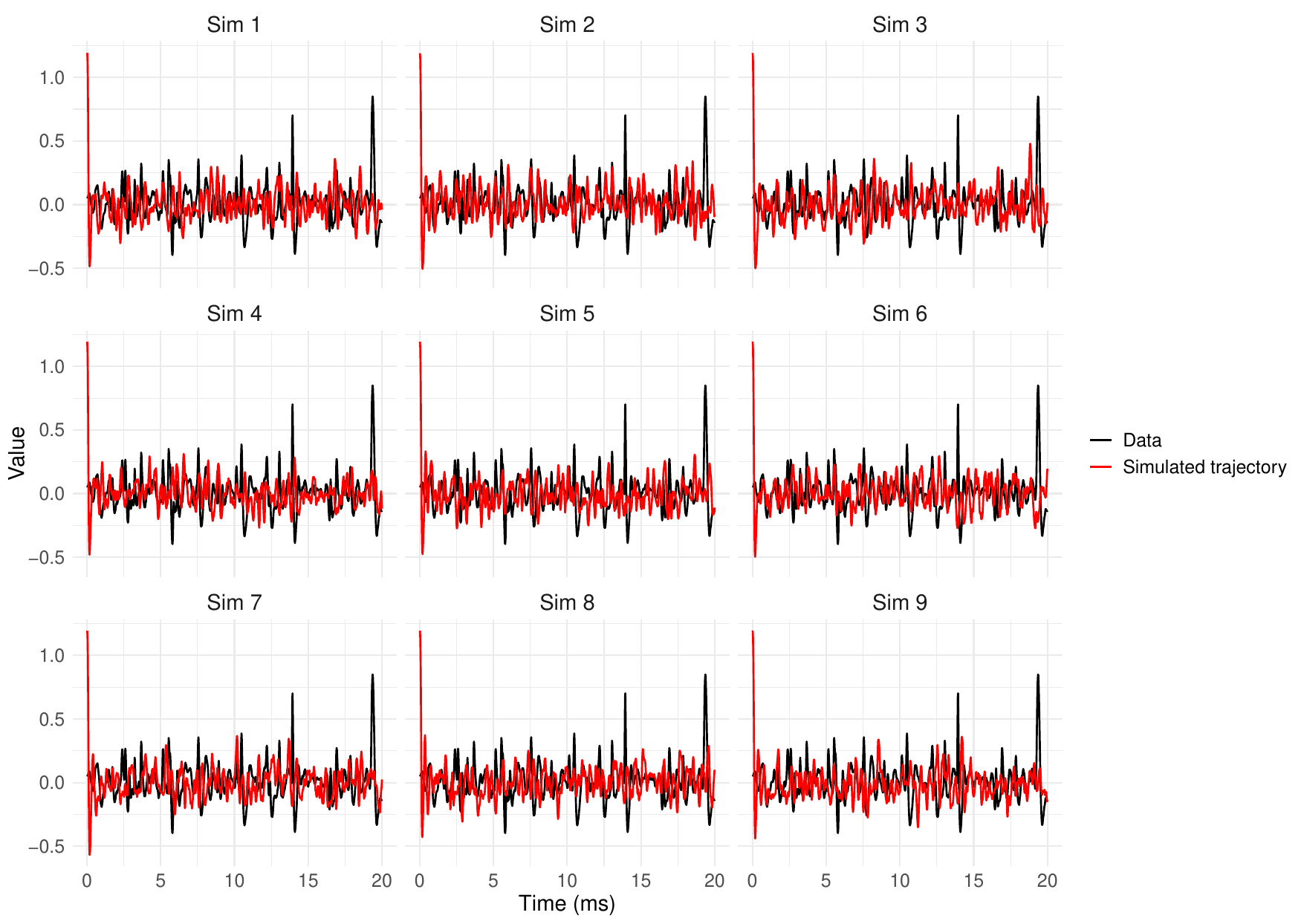}
  \caption{Simulated trajectories compared to observed data in 20 ms.}
  \label{sim_traj}
\end{figure}

\end{document}